\documentclass[reprint, onecolumn, showkeys, aps, pra, superscriptaddress]{revtex4-2}
\usepackage{amsthm}
\newtheorem{theorem}{Theorem}

\newtheorem{lemma}[theorem]{Lemma}

\newtheorem{corollary}[theorem]{Corollary}

\usepackage{tabularx}
\usepackage{booktabs}
\usepackage{fancyhdr}
\usepackage{ulem}
\usepackage{soul}
\usepackage{color}
\usepackage{graphicx}
\usepackage{enumitem}
\usepackage{amsfonts}
\usepackage{amssymb}
\usepackage{amsmath,latexsym,amsthm}
\usepackage{physics}
\usepackage{bm}
\usepackage{xcolor}
\usepackage{diagbox}
\usepackage{multirow}
\usepackage[bottom]{footmisc}
\usepackage[breaklinks=true,colorlinks=true,linkcolor=teal,urlcolor=teal,citecolor=teal]{hyperref}

\usepackage{tikz}
\usetikzlibrary{matrix,decorations.pathreplacing,calc,positioning,calligraphy}
\tikzstyle{block} = [rectangle, draw, fill=white, rounded corners,
                 minimum height=6em, minimum width =8em, text width=7em, text centered]
\tikzstyle{midblock} = [rectangle, draw, fill=white, rounded corners,
                 minimum height=6em, minimum width =11em, text width=10em, text centered]
\tikzstyle{wideblock} = [rectangle, draw, fill=white, rounded corners,
                 minimum height=6em, minimum width =13em, text width=12em, text centered]
\tikzstyle{xwideblock} = [rectangle, draw, fill=white, rounded corners,
                 minimum height=6em, minimum width =22em, text width=21em, text centered]
\tikzstyle{xxwideblock} = [rectangle, draw, fill=white, rounded corners,
                 minimum height=6em, minimum width =27em, text width=26em, text centered]
\tikzstyle{xxxwideblock} = [rectangle, draw, fill=white, rounded corners,
                 minimum height=6em, minimum width =40em, text width=37em, text centered]
\def\ket#1{ | #1 \rangle}

\newcommand{\AD}[1]{\textcolor{black}{#1}}

\begin{document}
\title{Quantum algorithms for solving a drift-diffusion equation: a complexity analysis}

\author{Ellen Devereux}
\email{Ellen.Devereux@warwick.ac.uk}
\affiliation{Department of Physics, University of Warwick, Coventry CV4 7AL, United Kingdom}
\affiliation{Fujitsu UK Ltd}


\author{Animesh Datta}
\affiliation{Department of Physics, University of Warwick, Coventry CV4 7AL, United Kingdom}

\date{\today}% It is always \today, today,

\begin{abstract}
%Quantum computing offers the potential to revolutionize the solution of complex real-world modelling problems. 
We present \AD{four} quantum algorithms for solving a \AD{multidimensional} drift-diffusion equation.
They rely on a quantum linear system solver, a quantum Hamiltonian simulation, a quantum \AD{random-walk}, and the quantum Fourier transform.
 We compare the complexities of these methods to their classical counterparts, finding that \AD{diagonalization} via the quantum Fourier transform 
 offers a quantum computational advantage for solving linear partial differential equations at a fixed final time. 
 We employ a \AD{multidimensional} amplitude estimation process to extract the full probability distribution from the quantum computer.
\end{abstract}


\maketitle

\section{Introduction}
\label{sec: intro}

Current computational limitations constrain numerous scientific and commercial applications, primarily in terms of size, processing speed, and cost. For example, solving partial differential equations (PDEs) consumes significant memory and computational resources~\cite{Montanaro2015QuantumMethods}. Common numerical methods necessitate discretizing a large, complex problem domain, requiring repeated retrieval of stored information for each calculation. Furthermore, progressing the solution at each point in the \AD{discretization} in parallel uses vast amounts of computational resource.

One instance of a computationally expensive PDE is the drift-diffusion equation (DDE). Its importance stems from its role as a Fokker-Planck equation describing particle velocity, and its close relationship to the Black-Scholes and Navier-Stokes equations. Furthermore, a DDE can describe a stochastic process such as the Ornstein-Uhlenbeck process. Therefore, efficient methods of solution to the DDE have value across multiple industries. It is thus of significant commercial interest.
PDEs like the DDE are used to assist decision making across industries,
for example, instances include 
modelling wind turbine power outputs \cite{Arenas-Lopez2020AOutput} or risk outcomes through the pricing of options within finance \cite{Montanaro2015QuantumMethods}. %Best practice for 
Methods for solving linear PDEs like the DDE use linear equation solvers for the finite difference method such as the conjugate gradient method \AD{\cite{Grossmann2007NumericalEquations,Shewchuk1994AnPain}}, \AD{random-walk}s \cite{Horn1985NormsMatrices}, or \AD{diagonalization} by the Fourier transform \cite{Wu2020Diagonalization-basedProblems}.

Quantum algorithms promise potential computational advantage for solving differential equations.
Previous works on quantum algorithms for solving PDEs  can be split into two groups.
The first maps the PDE to a Hamiltonian using either \AD{discretization} or a variational approach. The result can be solved via Schrodinger's equation~\cite{Alghassi2022AFormula, Jin2023QuantumAnalysis, Arrazola2019QuantumEquations, Leong2022VariationalSolver, Sato2025QuantumParameters, Jin2024QuantumSchrodingerization, Over2025QuantumOperator, Brearley2024QuantumSimulation}. 
The second uses \AD{discretization} to build a system of linear equations to be solved quantum mechanically~\cite{An2022AFast-forwarding}.
Both document a $1/\epsilon$ time complexity for various PDEs \cite{Jin2022TimeEquations, Berry2014High-orderEquations, Montanaro2016QuantumMethod, Tennie2024SolvingApproach, Ingelmann2024TwoEquation}, where $\epsilon$ is the error in the accuracy of solution. 

In this paper, we study four different quantum algorithms for solving the DDE in \AD{Eq.} (\ref{eqn: DDE}).
We seek the probability distribution $p(\mathbf{x}, t)$ up to a constant additive error $\epsilon$. 
To that end, we  employ a method of extracting a probability distribution approximately up to an error $\epsilon_q$ as defined in \AD{Eq.}~\eqref{eqn: def_epsq}~\cite{VanApeldoorn2021QuantumEstimation}.
We believe ours is its first application to solving PDEs.
It is in effect a \AD{multidimensional} amplitude estimation procedure to extract a probability distribution from a quantum computer. 
% However, this can only be done up to an error $\epsilon_q$ as defined in \AD{Eq.}~\eqref{eqn: def_epsq}.
Our work is inspired by an earlier work on solving the heat equation~\cite{Linden2022QuantumEquation}.
It sought the total heat in a region of space up to a constant additive error at a final time, where total heat is an integral of the solution to the heat equation. 
This difference in the objective we seek has several consequences which we discuss in Sec.~\ref{sec:Classical Methods} and Sec.~\ref{sec:Quantum Methods}.  

All of our classical algorithms and three of our quantum algorithms based on the quantum \AD{random-walk}, the quantum \AD{time-evolution} and the quantum Fourier transform also apply to the advection diffusion equation (ADE) where the drift coefficient $a<0$ from \AD{Eq.}~\ref{eqn: DDE}. Our algorithm that is based on the quantum linear system solver only applies to the ADE in the same restricted circumstances as it applies to the DDE to bound the condition number. Outside of these circumstances the condition number tends to infinity.
The ADE in one spatial dimension has been solved using a quantum linear system as well as a variational quantum algorithm~\cite{Ingelmann2024TwoEquation}.
They show, as we do that the quantum linear systems method depends on the condition number\AD{, whereas while} the variational algorithm is more efficient it requires an ansatz and is only able to provide a partial solution.
A linear combination of \AD{H}amiltonian simulation (LCHS) method for solving linear and \AD{nonlinear} PDEs was introduced in Ref.~\cite{Novikau2024QuantumSystems}. This was also applied to ADE in one spatial dimension and found to be near optimal~\cite{NovikauExplicitEquation}. 
In this case they find a linear scaling with time $T$ and logarithmic scaling with their truncation error $\varepsilon$. Our method considers $d$ dimensions, and in one dimension we achieve a quadratically improved dependence on $T$. The truncation error is not equivalent to our $\epsilon_c$ but rather $\zeta$, both variables are introduced in Sec.~\ref{ssec: statement}. The LCHS method is an improved scaling of $\zeta$ compared to our methods. The overall accuracy is not considered analytically by Ref.~\cite{NovikauExplicitEquation} but used as a test of the complexity scaling.

Our findings are \AD{summarized} in Table \ref{tab: results_sum_full}.
Our central result is that quantum computational advantage, in terms of time complexity, is possible by the quantum \AD{diagonalization} method when 
\begin{equation}
   \epsilon_q \geq \Tilde{O}\left(\dfrac{\epsilon_c^{d/4}d D^{d/2}}{\zeta^{d/4}L^d(aL+D)^{d/2}}\right)  \quad\text{and}\quad \epsilon = \epsilon_c + \epsilon_q,  
\end{equation}
where $\epsilon_c$ is the classical \AD{discretization} accuracy defined by Eq.~\eqref{eqn: classical_approx_definition}, $\epsilon_q$ is a quantum approximation accuracy defined by \AD{Eq.}~\eqref{eqn: def_epsq} and $\Tilde{O}(\cdot)$ omits logarithmic factors.

In terms of space complexity, each of the quantum methods has two main contributions. 
The first is $q$, the number of qubits required to build the state $|\tilde{p}\rangle$ as defined in \AD{Eq.}~\eqref{eqn: quantum_state}, where $q$ is $O(d\log n_x)$.
 The second is the number of qubits required to perform the measurement protocol to extract $\Tilde{\Tilde{\mathbf{p}}}$ such that $||\Tilde{\Tilde{\mathbf{p}}}- \mathbf{p}||_{\infty}\leq \epsilon$, which is $O(q + 1/\epsilon_q(\log(1/\epsilon_q)+q))$ \cite{VanApeldoorn2021QuantumEstimation}. Therefore, all of the quantum methods have a space complexity of $\Tilde{O}(d/\epsilon_q)$.

\begin{table}
    \centering
    \begin{tabularx}{0.75\textwidth}{X X X X X X}%{l l l l l l}
    \hline
    \textbf{Method}   & \textbf{Region} & \AD{\textbf{Theorem}} & \textbf{Complexity}\\
    \hline
    \textit{Classical} & & & & \\
    Linear \AD{e}quations &  General & \ref{thm: Lineqn} & $\Tilde{O}\left(\dfrac{d^{d/2+4}T^{d/2+3}\zeta^{(d+3)/2}L^{3+d}}{\epsilon_c^{(d+3)/2}}\right)$ \\ \\
    Time-\AD{s}tepping &  General & \ref{thm: timestepping} & $\Tilde{O}\left(\dfrac{d^{d/2+3}T^{d/2+2}\zeta^{d/2+1}L^{d+2}}{\epsilon_c^{d/2+1}}\right)$ \\ \\
    Exact sampling &  General & \ref{thm: exact_sampling} & $\Tilde{O}\left(\dfrac{d^{d/2+3}T^{d/2+2}\zeta^{d/2+1}L^{d+2}}{\epsilon_c^{d/2+1}}\right)$ \\ \\
    FFT &  Rectangular & \ref{thm: Classical_FFT} & $\Tilde{O}\left(\dfrac{d^{d/2+1}T^{d/2}\zeta^{d/2}L^{2d}}{\epsilon_c^{d/2}}\right)$ \\ \\
    \textit{Quantum} & & \\
    Linear equations &  General & \ref{thm: QLEM} & $\Tilde{O}\left(\dfrac{d^5 T^2 \zeta L^2}{\epsilon_q \epsilon_c} \right) $ \\ \\
    Time evolution &  General & \ref{thm: QTM} & $\Tilde{O}\left(\dfrac{d^{d/2+3} T^{d/2+2}\zeta^{d/4+1}L^{d/2+2}}{\epsilon_q\epsilon_c^{d/4+2}} \right) $ \\ \\
    Quantum RW
    &  General & \ref{thm: QRW} & $\Tilde{O}\left(\dfrac{d^{d/2 + 7/2}T^{d/2+1}\zeta^{d/4+1/2}L^{d/2+1}}{\epsilon_q\epsilon_c^{d/4+1/2}}  \right)$\\ \\
    QFT &  Rectangular & \ref{thm: QFT} & $\Tilde{O}\left( \dfrac{d^{(d/2+2)}T^{d/2}\zeta^{d/4}L^{d/2}}{ \epsilon_q\epsilon_c^{d/4} }\right)$\\ \\
    \hline
    \end{tabularx}
    \caption{Time complexity to solve the DDE from \AD{Eq.}~\eqref{eqn: DDE} in terms of parameters defined in Sec. \ref{sec: sol_strat}. We compare four classical and \AD{four} quantum methods. Logarithmic factors are omitted in $\Tilde{O}(\cdot)$.}
    \label{tab: results_sum_full}
\end{table}

% \dfrac{d^{d/2+3} T^{d/2+2}\zeta^{d/4+1}L^2(aL+D)^{d/2}}{\epsilon_q\epsilon_c^{d/4+2}}

This paper has the following structure\AD{: F}irst we state the problem and a summary of our findings in Sec. \ref{ssec: statement} and \ref{ssec: results_sum}, \AD{and} then we describe the preliminaries and solution strategy used across all solution methods in Sec. \ref{sec: sol_strat}.
We then identify the complexity of four classical solution methods in Sec. \ref{sec:Classical Methods}.
These are the conjugate gradient method for linear equations, linear \AD{time-evolution}, a random-walk method, and diagonalization via the fast Fourier transform.
These are compared with the quantum methods in Sec. \ref{sec:Quantum Methods} which use linear equation solvers, \AD{time-evolution} by Hamiltonian simulation, quantum \AD{random-walk} simulation, and \AD{diagonalization} by quantum Fourier transform (QFT). We conclude our findings in Sec.~\ref{sec: Conclusion}.

\subsection{Problem \AD{s}tatement}
\label{ssec: statement}
We aim to solve the \AD{DDE}
\begin{equation}
    \frac{\partial p(\mathbf{x}, t)}{\partial t}
    = \sum_{i = 1}^d \bigg[ a \frac{\partial}{\partial x_i} [p(\mathbf{x}, t)] + D \frac{\partial^2}{\partial x_i^2} [p(\mathbf{x}, t)]\bigg],
    \label{eqn: DDE}
\end{equation}
where $\mathbf{x} \equiv \{x_1,\AD{\ldots},x_d\}~\in \mathbb{R}^d$ is a  $d$-dimensional vector and $a$ and $D$ are positive constants representing the diffusion and drift coefficients respectively. 
We seek an approximate solution, $\Tilde{\Tilde{p}}(\mathbf{x},t)$, of $p(\mathbf{x},t)$ from \AD{Eq.}~\eqref{eqn: DDE} up to a given error $\epsilon \in (0,1)$ at a time $t = T$ by
\begin{equation}
    \big|\big|\Tilde{\Tilde{p}}(\mathbf{x}, t) - p(\mathbf{x}, t) \big| \big|_{\infty}\leq \epsilon 
    \label{eqn: approx_definition}
\end{equation}
for $\mathbf{x} \in [-L,L]^d$. The infinity norm is defined as $||\mathbf{x}||_{\infty}=\AD{\max}_j|x_j|$.

The solution $p(\mathbf{x}, t)$ is positive   
and $\int_{-\infty}^\infty p(\mathbf{x}, t) d\mathbf{x} = 1$ at all times \cite{Risken1996Fokker-PlanckEquation}. $p(\mathbf{x}, t)$ is also dimensionless.
We assume periodic boundary conditions in all spatial dimensions $x_j$ but not in time $t$, such that $p(L, t) = p(- L, t)$.  
We also assume that $p(\mathbf{x}, t)$ is four times differentiable implying that 
\begin{subequations}
    \begin{gather}
        \underset{(\mathbf{x}, t) \in \mathbb{R}^{d+1}}{\max}\Big|\frac{\partial ^4 p(x_1,...,x_d,t)}{\partial x_i^2 \partial x_j^2}\Big| \leq \zeta, \\
        \underset{(\mathbf{x}, t) \in \mathbb{R}^{d+1}}{\max}\Big|\frac{\partial ^3 p(x_1,...,x_d,t)}{\partial x_i \partial x_j \partial x_k}\Big| \leq \zeta L, \\
        \underset{(\mathbf{x}, t) \in \mathbb{R}^{d+1}}{\max}\Big|\frac{\partial ^2 p(x_1,...,x_d,t)}{\partial x_i \partial x_j}\Big| \leq \zeta L^2, \\
        \underset{(\mathbf{x}, t) \in \mathbb{R}^{d+1}}{\max}\Big|\frac{\partial p(x_1,...,x_d,t)}{\partial x_i}\Big| \leq \zeta L^3\AD{,}
    \end{gather}
    \label{eqn: smoothness_bound}
\end{subequations}
with a smoothness bound $\zeta$ with dimensions of $\mathrm{(length)}^{-4}$. 

We also make the following assumptions on the computational costs:
\begin{enumerate}
    \item For classical computation
    \begin{enumerate}
        \item elementary arithmetic operation (addition or multiplication) on real numbers can be computed in constant time,
        \item generation of a real random number in $[0, 1]$ can be completed in constant time,
        \item $p(\mathbf{x}, 0) = p_0(\mathbf{x})$ and its powers can be computed exactly at no cost for all $\bf{x}$.
    \end{enumerate}
    \item For quantum computation:
    \begin{enumerate}
        \item  that one and two qubit gates can be performed in constant time.
    \end{enumerate} 
 \end{enumerate}
Since $O(\log_2(x))=O(\log_e(x))$ we will not discern between logarithms of different bases.

\subsection{Results summary}
\label{ssec: results_sum}

We show that quantum computers can provide computational advantage in solving the DDE in \AD{Eq.}~\eqref{eqn: DDE} depending on the problem parameters. Specifically we focus on the time complexity of a few solution methods. 
As described more fully in Sec. \ref{sec: intro} and \AD{summarized} in Table \ref{tab: results_sum_full}, the different methods we study are as follows:
\begin{enumerate}[label=\roman*]
    \item System of linear equations is a method for finding the solution at all points in space and time. It employs the conjugate gradient method to solve the system.
    \item Time evolution is the simplest and most efficient classical method for finding the solution at all points in space and time. 
    \item Random walk uses the stochastic nature of the equation to model each time step as a step of the \AD{random-walk}. 
    This is less efficient than \AD{time-evolution} due to the number of samples required to achieve the accuracy prescribed in \AD{Eq.}~\eqref{eqn: approx_definition}.
    \item \AD{Diagonalization} is the most efficient of the classical methods for solving \AD{Eq.}~\eqref{eqn: DDE} at a fixed final time $T$. It does this by \AD{utilizing} the fast Fourier transform to \AD{diagonalize} the differential operator.    
\end{enumerate}
The quantum methods we study are as follows:
\begin{enumerate}[label=\roman*]
    \item Quantum linear equations method is more efficient than the classical linear equations method and the classical \AD{time-evolution} method. This makes it the most efficient method for finding the solution at all space and time steps. 
    \item Quantum \AD{time-evolution} method uses Hamiltonian simulation and the linear combination of unitaries methods to evolve the solution with time. This method uses work from Ref. \cite{Over2025QuantumOperator}. It is less efficient than classical time stepping.
    \item  Quantum \AD{random-walk} method takes advantage of the stochastic nature of \AD{random-walk}s to efficiently model each time step as a step of a quantum \AD{random-walk}. This method uses work from Ref. \cite{Apers2018QuantumTesting}. It is more efficient than the above methods but not the most efficient. % method. 
    \item Quantum \AD{diagonalization} uses the quantum Fourier transform to \AD{diagonalize} the differential operator. This is the most efficient method for solving \AD{Eq.}~\eqref{eqn: DDE} for a fixed final time $T.$
\end{enumerate}

\subsection{In practice}
This section will demonstrate the results quoted in Table \ref{tab: results_sum_full} using a commercially relevant use case. Several works have modelled stock volatility as an Ornstein-Uhlenbeck process, which is a specific instance of a DDE \cite{Andersen2001TheVolatility, Wang2023ModelingProcess}. Based on these works typical values for the parameters of the DDE are shown in Table \ref{tab: variables}. We chose coefficients $a$ and $D$ from Ref. \cite{Wang2023ModelingProcess}. In this case $a$ represents the mean reversion parameter, where mean reversion is the theory that asset prices eventually revert to their long-term mean. 
$D = \sigma^2/2$ where $\sigma$ is the volatility of the stock, as volatility itself is an \AD{annualized} standard deviation. 
$T$ is chosen based on the number of samples used by Ref. \cite{Wang2023ModelingProcess} and represents time as expected. $L$ is chosen based on the spread of volatility results. Stock volatility is highly dimensional as it is dependent on many variables \cite{Andersen2001TheVolatility}. We chose $d=3$ for this exercise because volatility depends minimally on three factors: time to maturity, strike price, and market conditions.

We demonstrate the quantum computational advantage by substituting the variables from Table \ref{tab: variables} into the classical and quantum method complexities from Table \ref{tab: results_sum_full}. These results are shown in Table \ref{tab: in_practice}. To obtain the required overall accuracy of $\epsilon$ this comparison can be carried out in advance to select suitable values for $\epsilon_c$ and $\epsilon_q$.

\begin{table*}
\begin{minipage}[b]{.45\linewidth}
  \centering
    \begin{tabular}{ c c c c c c }
        \hline 
        \vspace{2pt}
         $T$ [days] & $L$ [$\$$] & $a$ & $D$ & $d$ & $\zeta$ [$\$^{-4}$] \\ 
         \hline \\
         \vspace{2pt}
         5000 & 10 & 0.2366 & 0.2455 & 3 & 1 \\  \\
         \hline
    \end{tabular}
    \caption{Parameter values for DDE for a representative problem in finance \cite{Andersen2001TheVolatility, Wang2023ModelingProcess}. These variables result in $n_t = 8.13\times 10^5/\epsilon_c$ and $n_x = 1.05 \times 10^2/\epsilon_c^{1/2}$ using \AD{Corollary} \ref{cor:dxdt bounds}.
    \label{tab: variables}}
\end{minipage}%
\hspace{2em}
\begin{minipage}[b]{.45\linewidth}
  \centering
    \begin{tabular}{l l l l }
        \hline
        Method    & Theorems  &  \multicolumn{2}{c}{Complexity} \\
        && Classical & Quantum  \\
        \hline
        \\
        Linear \AD{e}quations &  \ref{thm: Lineqn} \AD{and} \ref{thm: QLEM} & $\dfrac{4.85\times10^{22}}{\epsilon_c^{3}}$ & $\dfrac{4.14\times10^{10}}{\epsilon_c \epsilon_q}$ \\ \\
        Time \AD{s}tepping &  \ref{thm: timestepping} \AD{and} \ref{thm: QTM} & $\dfrac{1.24\times10^{18}}{\epsilon_c^{2.5}}$&$\dfrac{5.23\times10^{17}}{\epsilon_c^{2.75} \epsilon_q}$ \\ \\
        Random \AD{w}alk &  \ref{thm: exact_sampling} \AD{and} \ref{thm: QRW} & $\dfrac{1.24\times10^{18}}{\epsilon_c^{2.5}}$ & $\dfrac{4.73\times10^{12}}{\epsilon_c^{1.25} \epsilon_q}$  \\ \\
        \AD{Diagonalization} &  \ref{thm: Classical_FFT} \AD{and} \ref{thm: QFT} & $\dfrac{8.07\times10^{11}}{\epsilon_c^{1.5}}$ & $\dfrac{6.98\times10^7}{\epsilon_c^{0.75} \epsilon_q}$     \\ \\
        \hline
    \end{tabular}
    \caption{Summary of results with respect to the chosen variables from \AD{Table} \ref{tab: variables} for solving the DDE. 
    \label{tab: in_practice}}
\end{minipage}
\end{table*}

\section{Preliminaries and solution strategy}
\label{sec: sol_strat}
\subsection{Preliminaries}
\label{sec: prelim}
In this section we list the notation used throughout this paper. All vectors will be identified using lower case \AD{bold font}, such as $\mathbf{x}$ representing a $d$-dimensional vector in space and individual components are denoted by $x_j$.
For the  vectors, we use the following  norms: 
the infinity norm defined as $||\mathbf{x}||_{\infty}=\AD{\max}_j|x_j|$, 
the one norm defined as $||\mathbf{x}||_1 = \sum_{j=1}^{n} |x_j|$, 
the Euclidean norm defined as $||\mathbf{x}||_2 = \sqrt{x_1^2+...+x_n^2}$ \cite{Horn1985NormsMatrices}, 
the operator norm $||\cdot|| = \underset{j}\max     ~\sigma_j$, where $\sigma_j$ denotes the $j$th singular value,
and the energy norm $||\mathbf{x}||_{\mathcal{A}}$ with respect to a positive semi-definite matrix $\mathcal{A}$ defined as $||\mathbf{x}||_{\mathcal{A}} = \sqrt{\mathbf{x}^T\mathcal{A}~\mathbf{x}}$.

All matrices are denoted by an \AD{uppercase} Latin alphabet, such as $\mathcal{A}$.  In particular, $\mathcal{I}_m$ denotes the identity matrix of dimension $m \times m.$ 
The condition number $\kappa_\mathcal{A}$ a matrix $\mathcal{A}$ is defined as $\kappa_\mathcal{A} = ||\mathcal{A}||\, ||\mathcal{A}^{-1}||,$ noting the use of the subscript to denote ownership. Likewise where possible, their eigenvalues will be denoted with a lowercase Greek character corresponding to the matrix 
e.g., $\alpha_j$ denotes the $j$th eigenvalue of the matrix $\mathcal{A}$.

The solution is $p(\mathbf{x},t)$ at time $t$ and position $\mathbf{x}$.
All our solutions employ \AD{discretization} on a grid referred to as $G$, which spans $[-L,L]^d$ in space and $[0,T]$ in time.
This grid is \AD{discretized} into $n_x$ equally spaced points in each of the $d$ spatial dimensions ($n_x$ must be even) and $n_t$ equally distributed points in time. 
The spacing of these points is $\Delta x = 2L/n_x$ in all space dimensions and $\Delta t = T/n_t$.
Therefore the grid $G$ is described by the discrete points $(\mathbf{x}, t)$ which, per dimension $x_j$ are\AD{,} 
$x_{j[-n_x/2]} = -L = -n_x\Delta x/2,
\AD{\ldots}, x_{j[0]} = 0,
\AD{\ldots}, x_{j[n_x]} = n_x\Delta x/2 = L$; $ t_0 = 0, t_1 = \Delta t, \AD{\ldots}, t_{n_t} = n_t\Delta t = T$. 

The \AD{discretized} approximation of the solution is denoted by  $\tilde{p}(\mathbf{x},t)$.
Alternatively, we will use $\tilde{\mathbf{p}}_k = \tilde{p}(\mathbf{x}, t = k\Delta t)$ at some points to denote the vector of solutions. $\tilde{\mathbf{p}}_k$ is an $n_x^d \times 1$ vector.
The error introduced by this \AD{discretization} is denoted as $\epsilon_c$, 
\begin{equation}
    \big|\big|\tilde{p}(\mathbf{x}, t) - p(\mathbf{x}, t) \big| \big|_{\infty}\leq \epsilon_c \AD{,}
    \label{eqn: classical_approx_definition}
\end{equation}
the $c$ subscript here is used since \AD{discretization} produces an error that is entirely classical in nature. 
The \AD{discretization} method in both space and time that results in $\epsilon_c$ is described in Sec.~\ref{ssec: discr}.

For the quantum methods, we use Dirac notation such that a quantum state is represented as a ket, $|b\rangle$ and its Hermition conjugate is the bra, $\langle b|$. 
The approximate probability distribution $\tilde{p}(\mathbf{x}, t)$ is encoded in the quantum state 
\begin{equation}
        |\tilde{p}\rangle = \frac{1}{||\tilde{p}(\mathbf{x}, t)||_{2}}\sum_{(\mathbf{x},t)\in G} \tilde{p}(\mathbf{x},t)|\mathbf{x},t\rangle.
        \label{eqn: quantum_state}
    \end{equation}   
As we can only measure the above quantum state to a finite error,  the extracted approximated probability distribution $\tilde{\tilde{p}}(\mathbf{x}, t)$ is such that
     \begin{equation}
        ||\tilde{\tilde{p}}(\mathbf{x}, t) - \tilde{p}(\mathbf{x}, t)||_{\infty}\leq \epsilon_q\AD{.}
        \label{eqn: def_epsq}
    \end{equation}
The error in approximating the quantum state is thus $\epsilon_q,$ and the overall error is $\epsilon = \epsilon_c + \epsilon_q$ as in \AD{Eq.} (\ref{eqn: approx_definition}). 

\subsection{Solution strategy}
\label{ssec: discr}
All our methods rely on a system of linear equations built via \AD{discretization} 
of \AD{Eq.} (\ref{eqn: DDE}) using Taylor's theorem. The forward-time, centered-space \AD{discretization} method in one spatial and one temporal dimension results in     
\begin{equation}
    \frac{dp}{dt} = \frac{p(t + \Delta t) - p(t)}{\Delta t} - \frac{\Delta t}{2}\frac{d^2p(\xi)}{dt^2},
    \label{eqn: TTpt}
\end{equation}

\begin{equation}
    \frac{dp}{dx} = \frac{p(x + \Delta x) - p(x - \Delta x)}{2\Delta x} - \frac{\Delta x^2}{12}\bigg(\frac{d^3p(\xi')}{dx^3} + \frac{d^3p(\xi'')}{dx^3} \bigg),
    \label{eqn: TTpx}
\end{equation}

\begin{equation}
    \frac{d^2p}{dx^2} = \frac{p(x + \Delta x) + p(x - \Delta x) - 2p(x)}{\Delta x^2} - \frac{\Delta x^2}{24}\bigg(\frac{d^4p(\xi')}{dx^4} + \frac{d^4p(\xi'')}{dx^4} \bigg),
    \label{eqn: TTpxx}
\end{equation}
where $\xi \in [t, t + \Delta t], \xi' \in [x, x + \Delta x]$ and $ \xi'' \in [x - \Delta x, x]$. 
These can be rearranged into
\begin{equation}
    \bigg|\frac{dp}{dt}- \frac{p(t + \Delta t) - p(t)}{\Delta t}\bigg| \leq \frac{\Delta t}{2} \underset{(\mathbf{x}, t) \in \mathbb{R}^{d+1}}{\max}\bigg|\frac{d^2p(t)}{dt^2}\bigg|\AD{,} \label{eqn: pt disc}
\end{equation}

\begin{equation}
    \bigg|\frac{dp}{dx} - \frac{(p(x + \Delta x) - p(x - \Delta x))}{2\Delta x} \bigg| \leq \frac{\Delta x^2}{6} \underset{(\mathbf{x}, t) \in \mathbb{R}^{d+1}}{\max}\bigg|\frac{d^3p(x)}{dx^3}\bigg|\AD{,}
    \label{eqn: px disc}
\end{equation}
\begin{equation}
    \bigg|\frac{d^2p}{dx^2} - \frac{p(x + \Delta x) + p(x - \Delta x) - 2p(x)}{\Delta x^2} \bigg| \leq \frac{\Delta x^2}{12} \underset{(\mathbf{x}, t) \in \mathbb{R}^{d+1}}{\max}\bigg|\frac{d^4p(x)}{dx^4}\bigg|.\label{eqn: pxx disc}
\end{equation}
Then we can use the smoothness bounds from \AD{Eq.} (\ref{eqn: smoothness_bound}) to create a smoothness bound on $ \underset{(\mathbf{x}, t) \in \mathbb{R}^{d+1}}{\max}\big|\frac{\partial^2 p(\mathbf{x},t)}{\partial t^2}\big|$. 
This requires first differentiating \AD{Eq.} (\ref{eqn: DDE}) by $t$ and then substituting the smoothness bounds in $\mathbf{x}$ into the result.  Therefore, the bound on $  \underset{(\mathbf{x}, t) \in \mathbb{R}^{d+1}}{\max}\big|\frac{\partial^2 p(\mathbf{x},t)}{\partial t^2}\big|$ is   
\begin{equation}
    \underset{(\mathbf{x}, t) \in \mathbb{R}^{d+1}}{\max}\bigg|\frac{\partial^2 p(\mathbf{x},t)}{\partial t^2}\bigg| \leq d^2 \zeta(aL + D)^2.
\end{equation}
Considering the points in $G$ and using the inequalities in \AD{Eqs.} \AD{\eqref{eqn: pt disc}, \eqref{eqn: px disc} and \eqref{eqn: pxx disc}}, the full \AD{discretized} version of \AD{Eq.} (\ref{eqn: DDE}) on the grid $G$ is 
\begin{multline}
    \frac{\tilde{p}(\mathbf{x}, t + \Delta t) - \tilde{p}(\mathbf{x}, t)}{\Delta t} =
     \sum_{j = 1}^d \bigg ( \frac{a }{2\Delta x}(\tilde{p}(...,x_{j[i]} + \Delta x,...,t) - \tilde{p}(...,x_{j[i]} - \Delta x,...,t)) + \frac{D}{\Delta x^2}(\tilde{p}(...,x_{j[i]} + \Delta x,...,t)\\
    + \tilde{p}(...,x_{j[i]} - \Delta x,...,t) - 2\tilde{p}(\mathbf{x},t)),
    \bigg).
    \label{eqn: disc OU}
\end{multline}
where $i = [-n_x/2,  \cdots, n_x/2]$ and $x_{j[i]}$ denotes
position $[i]$ in dimension $j$. Likewise, $x_{j[i]} + \Delta x$, denotes the step forward by $\Delta x$ from the current position, $[i]$ in dimension $j$. 
\AD{Equation} (\ref{eqn: disc OU}) can be rearranged as \begin{align}
        \tilde{p}(\mathbf{x}, t + \Delta t)  &= \mathcal{L}\tilde{p}(\mathbf{x}, t)  \nonumber\\ \nonumber
      &=  	\bigg(1 - \frac{2d D\Delta t}{\Delta x^2}\bigg) \tilde{p}(\mathbf{x}, t) + \Delta t \sum_{j = 1}^d \bigg ( \frac{a }{2\Delta x}(\tilde{p}(...,x_{j[i]} + \Delta x,...,t) - \tilde{p}(...,x_{j[i]} - \Delta x,...,t))
        \\&+ \frac{D}{\Delta x^2}(\tilde{p}(...,x_{j[i]} + \Delta x,...,t)
    + \tilde{p}(...,x_{j[i]} - \Delta x,...,t))
        \bigg)
        \label{eqn: diff operator}
    \end{align}
to define $\mathcal{L}$ as a linear operator of dimension $n_x^d \times n_x^d$.
This \AD{discretization} enables us to build a matrix representation of an approximate solution to the DDE in \AD{Eq.}~\eqref{eqn: DDE} at all points in $G$ via  $\tilde{\mathbf{p}}_{k+1} =\mathcal{L}\tilde{\mathbf{p}}_{k}$ for $k = 1, 2,..., n_t$. Then the solution is formally given by solving the linear system
\begin{equation}
    \mathcal{A}\tilde{\mathbf{p}} =
    \begin{pmatrix} 
       \mathcal{I}_{n_x^d} & & & \\
        -\mathcal{L} & \mathcal{I}_{n_x^d} & & \\
        & \ddots& \ddots& \\
        & & -\mathcal{L} & \mathcal{I}_{n_x^d} 
    \end{pmatrix} 
    \begin{pmatrix} 
       \tilde{\mathbf{p}}_1 \\
        \tilde{\mathbf{p}}_2 \\
        \vdots \\ 
        \tilde{\mathbf{p}}_{n_t} 
    \end{pmatrix} =
    \begin{pmatrix} 
        \mathcal{L} \tilde{\mathbf{p}}_0\\
        0 \\ 
        \vdots \\
        0   
    \end{pmatrix}\AD{,}
    \label{eqn: syslineareqns}
\end{equation}
where $\mathcal{A}$ is an $n_x^dn_t \cross n_x^dn_t$ matrix. We will describe the steps to build the constituent matrices of $\mathcal{L}$ based on the \AD{discretization} given in \AD{Eq.}~\eqref{eqn: diff operator} in Sec.~\ref{sec: tech_ing}.

First, we use \AD{Eq.}~\eqref{eqn: diff operator} to define the size of the \AD{discretization} steps required  to achieve the approximation accuracy stated in \AD{Eq.} (\ref{eqn: approx_definition}), as shown in the following theorem.    
\begin{theorem}[Approximation up to $\infty$-norm error] If $\Delta t \leq  \Delta x^2/2dD$, \AD{then}
    \begin{equation}
        ||\tilde{\mathbf{p}}_{n_t} - \mathbf{p}_{n_t}||_\infty \leq n_t \frac{d \Delta t \zeta}{2}\bigg( d\Delta t (aL+D)^2 + \frac{\Delta x^2}{3} \Big(a L + \frac{D}{2}\Big) \bigg) = T \frac{d \zeta}{2}\bigg( d\Delta t (aL+D)^2 + \frac{\Delta x^2}{3} \Big(a L + \frac{D}{2}\Big) \bigg).
    \end{equation}
 
\label{thm: approx linf}
\end{theorem}
The proof is given in \AD{Appendix} \ref{app: strategy} %. In the proof of this theorem 
where we ensure $\mathcal{L}$ is stochastic. We do this by %implementing the 
bounding $\Delta t \leq  \Delta x^2/2dD$. Based on the bound from \AD{Theorem} \ref{thm: approx linf}, we define the size of $\Delta x$ and $\Delta t$ in the following corollary.
\begin{corollary}
    To approximate $p(\mathbf{x},t)$ up to $\infty$-norm error in $\epsilon$, it is adequate to take
    \begin{equation}
        \Delta x \leq \sqrt{\frac{12 \epsilon D}{T d \zeta (3a^2L^2 + 8aDL + 4D^2)}} \quad\text{and}\quad \Delta t \leq \frac{6 \epsilon }{T d^2\zeta (3a^2L^2 + 8aDL + 4D^2)}\AD{,}
    \end{equation}
    which correspond to %s to using 
    \begin{equation}
        n_t = \frac{T^2 d^2\zeta (3a^2L^2 + 8aDL + 4D^2)}{6 \epsilon},
    \end{equation} 
    and 
    \begin{equation}
        n_x = \sqrt{\frac{T d \zeta L^2(3a^2L^2 + 8aDL + 4D^2)}{3 \epsilon D}}.
    \end{equation}
    \label{cor:dxdt bounds}
\end{corollary}

\begin{proof}
    Take $ \Delta t = \Delta x^2/2 dD$ and set
    \begin{equation}
        ||\mathbf{\tilde{p}}_{n_t} - \mathbf{p}_{n_t}||_\infty \leq  T \frac{d \zeta}{2}\bigg( d\Delta t (aL+D)^2 + \frac{\Delta x^2}{3} \Big(a L + \frac{D}{2}\Big) \bigg)= \epsilon.
    \end{equation}
Substitute $\Delta t$ and rearrange for $\Delta x^2$. 
\end{proof} 
Throughout the rest of this paper we will use the values of $n_t$ and $n_x$ from Corollary \ref{cor:dxdt bounds}. 
Whil\AD{e} these inequalities must be met, the terms can be traded off against each other to achieve the necessary accuracy for a given problem.
Finally, the \AD{normalization} used is
\begin{equation}
    ||\mathbf{p}_0||_1 = \sum_{(\mathbf{x}_0, 0) \in G}p_0(\mathbf{x})  =1\AD{.} 
    \label{eqn: norm_p0}
\end{equation}
This captures  $\int_{[-L,L]^d} p_0(\mathbf{x})dx_1...dx_d = 1.$
Then by the stochasticity of $\mathcal{L}$ that we ensured during the proof of \AD{Theorem} \ref{thm: approx linf}, it is implied that $||\tilde{\mathbf{p}}_k||_1 = 1 $ for all $k$.

\section{Classical Algorithms}
\label{sec:Classical Methods}
\subsection{Technical ingredients}
\label{sec: tech_ing}

In this section we present the technical ingredients we use in the rest of Sec.~\ref{sec:Classical Methods}.
The main one is the condition number of the matrix $\mathcal{A}$ in \AD{Eq.} (\ref{eqn: syslineareqns}).  
It is key to determining the complexity of both the classical and quantum algorithms. 
Another is the singular values of $\mathcal{A}$ and the eigenvalues of certain constituent matrices of $\mathcal{A}.$

To determine the condition number of matrix $\mathcal{A}$, we first construct it from its constituent matrices. This involves revisiting the discretization in \AD{Eq.}~\eqref{eqn: diff operator} and representing each \AD{discretization} step using matrices.
To begin, we ignore the \AD{discretization} in time %for now 
and consider only the \AD{discretization} of space in one dimension 
\begin{equation}
\label{eq:1d}
    \frac{d\Tilde{p}(x, t)}{dt} = \sum_{x = -L}^{x = L} \frac{a} {2\Delta x}\Big(\Tilde{p}(x + \Delta x, t) - \Tilde{p}(x - \Delta x,t)\Big) + \frac{D}{\Delta x^2}\Big(\Tilde{p}(x + \Delta x,t) + \Tilde{p}(x - \Delta x,t) - 2\Tilde{p}(x,t)\Big).
\end{equation}
This can be written as 
\begin{equation}
    \frac{d\tilde{\mathbf{p}}}{dt} = \mathcal{H} \tilde{\mathbf{p}}\AD{,}
    \label{eqn: 1D disc}
\end{equation}
where
\begin{equation}
    \mathcal{H} = 
    \begin{pmatrix}
        -2\frac{D}{\Delta x^2} & \frac{D}{\Delta x^2} + \frac{a}{2\Delta x} & 0 & \dots &0& \frac{D}{\Delta x^2} - \frac{a}{2\Delta x}\\
        \frac{D}{\Delta x^2} - \frac{a}{2\Delta x} & -2\frac{D}{\Delta x^2} & \ddots & & & 0\\
        0&\ddots &\ddots & & & \vdots\\
        \vdots& & & & & 0\\
        0 & & & & & \frac{D}{\Delta x^2} + \frac{a}{2\Delta x} \\
        \frac{D}{\Delta x^2} + \frac{a}{2\Delta x} & 0 & \dots &0 & \frac{D}{\Delta x^2} - \frac{a}{2\Delta x} & -2\frac{D}{\Delta x^2}    
    \end{pmatrix}.
    \label{eqn: matrix_H}
\end{equation}
Since periodic boundary conditions are assumed in Sec. \ref{ssec: discr}, $\mathcal{H}$  is a circulant matrix of dimension $n_x \times n_x$.

Next we consider the \AD{discretization} in $t$ and extend \AD{Eq.}~\eqref{eqn: 1D disc} to $d$ dimensions to give
\begin{equation}
    \mathcal{\mathcal{A}}\tilde{\mathbf{p}} =
    \begin{pmatrix} 
        \mathcal{I}_{n_x^d} & & & \\
        -(\mathcal{I}_{n_x^d} + \Delta t \mathcal{M}) & \mathcal{I}_{n_x^d} & & \\
        & \ddots& \ddots& \\
        & & -(\mathcal{I}_{n_x^d} + \Delta t \mathcal{M}) & \mathcal{I}_{n_x^d} 
    \end{pmatrix} 
    \begin{pmatrix} 
       \tilde{\mathbf{p}}_1 \\
        \tilde{\mathbf{p}}_2 \\
        \vdots \\ 
        \tilde{\mathbf{p}}_{n_t} 
    \end{pmatrix} =
    \begin{pmatrix} 
        \mathcal{L}\tilde{\mathbf{p}}_0 \\
        0 \\ 
        \vdots \\
        0   
    \end{pmatrix}\AD{,}
\end{equation}
 where
\begin{equation}
    \mathcal{M} = \sum_{j = 1}^d\mathcal{I}_{n_x}^{\otimes (j-1)}\otimes \mathcal{H} \otimes \mathcal{I}_{n_x}^{\otimes (d-j)}\AD{.}
    \label{eqn: matrix_M}
\end{equation}
Then the differential operator $\mathcal{L}$ from \AD{Eq.} (\ref{eqn: diff operator}) is
\begin{equation}
    \mathcal{L}= \mathcal{I}_{n_x^d} + \Delta t \mathcal{M}.
    \label{eqn: M to L}
\end{equation}
Finally, the full system of linear equations described in \AD{Eq.} (\ref{eqn: syslineareqns}) is denoted as 
\begin{equation}
    \mathcal{A} = \mathcal{T} \otimes \mathcal{L} - \Delta t \mathcal{I}_{n_t} \otimes \mathcal{M},
\end{equation} where 
 \begin{equation}
\mathcal{T} =  \left(\begin{matrix}
    1 & & & \\
    -1 & 1 & & \\
    &\ddots &\ddots & \\
    & & -1 & 1
\end{matrix}\right)
\end{equation}
 is an $n_t \times n_t$ matrix.     

To find the condition number of $\mathcal{A}$ we require its singular values, which can be found via their relation to the eigenvalues of its constituent matrices $\mathcal{M}$ and $\mathcal{L}$. The eigenvalues of $\mathcal{M}$ are denoted as $\mu_j$ and given in Lemma \ref{lem: eigen_M}. 
Then by \AD{Eq.} (\ref{eqn: M to L}) the eigenvalues of $\mathcal{L}$ are 
\begin{equation}
    l_j = 1 + \Delta t \mu_j.
    \label{eqn: eigen_L}
\end{equation}

\begin{lemma}
    The eigenvalues of $\mathcal{M}$ are $\{\mu_{j_1}+ ... + \mu_{j_d}: j_1,...,j_d \in \{0,1,...,n_x-1\}\}$, where
    \begin{equation}
        \mu_j = -4\frac{D}{\Delta x^2} \sin^2\left(\frac{\pi j}{n_x}\right) + \AD{i} \frac{a }{\Delta x} \sin\left(\frac{2\pi j}{n_x}\right)
        \label{eqn: eigen_M}
    \end{equation}
    and $\AD{i}$ denotes the imaginary unit. Moreover, $\mathcal{M}$ is \AD{diagonalize}d by the $d$\AD{th} tensor product of the Fourier transform.
    \label{lem: eigen_M}
\end{lemma}
We provide a proof in \AD{Appendix} \ref{app: eigenvalues}.

Given the eigenvalues of $\mathcal{M}$ and $\mathcal{L}$ we can find the singular values of $\mathcal{A}$ using the $d$\AD{th} tensor power of the $n_x \times n_x$ Fourier transform $\mathcal{F}$. This follows from
\begin{equation}
    (\mathcal{I}_{n_t} \otimes \mathcal{F}^{\dag})\mathcal{A}(\mathcal{I}_{n_t}\otimes \mathcal{F}) = 
    \sum_{j=0}^{n_x-1}\mathcal{A}_j\otimes |j \rangle\langle j |
    \label{eqn: A transformation}
\end{equation}
showing that the singular values of $\mathcal{A}$ are the collection of singular values of 
\begin{equation}
    \mathcal{A}_j = (1 + \Delta t \mu_j) \mathcal{T} -\Delta t \mu_j \mathcal{I}_{n_t}
\end{equation}
for all $j \in [0, n_x-1].$
$\mathcal{A}_j$ is a nonsingular $n_t \times n_t$ matrix. 
To find the operator norm and condition number of the matrix $\mathcal{A}_j$, we need its maximum and minimum singular values, $\sigma_j$. $\sigma_j = \sqrt{\alpha_j}$, where $\alpha_j$ are the eigenvalues of $\mathcal{A}_j\,\mathcal{A}_j^{\dag}$. 
The eigenvalues of $\mathcal{A}$ and $\mathcal{A}_j$ are equivalent, as the transformation on the \AD{left-hand side} of \AD{Eq.} (\ref{eqn: A transformation}) is a unitary. 
Therefore, we need to find the eigenvalues of 
\begin{equation}
    \mathcal{A}_j\mathcal{A}_j^{\dag} = \begin{pmatrix}
        1 & -l_j^{\dag} & 0 & ... & 0 \\
        -l_j & 1 + |l_j|^2  & -l_j^{\dag} &  \\
        0 &  -l_j & 1 + |l_j|^2  & -l_j^{\dag} & \\
        \vdots & & -l_j & \ddots & \vdots \\
         & & \ddots & \ddots & -l_j^{\dag}  \\
        0 &  & ... & -l_j  & 1 + |l_j|^2 
    \end{pmatrix}\AD{,}
\end{equation}
where $l_j = 1 + \Delta t\mu_j = 1 - \Delta t \bigg(\frac{4 D}{\Delta x^2}\sin^2\left(\frac{\pi j}{n_x}\right) - \frac{ai}{\Delta x}\sin\left(\frac{2\pi j}{n_x}\right)\bigg).$

\begin{lemma}
    The eigenvalues of $\mathcal{A}_j\,\mathcal{A}_j^{\dag}$ take the form
    \begin{equation}
        \alpha_j = 1 + |l_j|^2 + 2|l_j|\cos \theta = \bigg( \frac{\sin \theta}{\sin (n_t \theta)}\bigg)^2\AD{,}
    \end{equation}
    where $\theta $ is defined by
    \begin{equation}
       |l_j| \sin (n_t \theta) + \sin((n_t+1)\theta) = 0
    \end{equation}
    and $\theta \neq k \pi$ for $k \in \mathbb{N}$.
    \label{lem: eigen_AAdag}
\end{lemma}
We provide the proof in \AD{Appendix} \ref{app: eigenvalues}.

To find the condition number of $\mathcal{A}$ we need the maximum and minimum $\sigma_j = \sqrt{\alpha_j}$. 
\begin{theorem}
    Condition number of $\mathcal{A}$ is
    \[ \kappa_\mathcal{A} = 
        \begin{cases}
            \Theta(n_t) \ \text{if}\ \dfrac{a^2 T}{2n_tD} \leq 1, \\ \\
            \Theta \bigg(\sqrt{n_t^2 + \dfrac{n_ta^2T}{D}}\bigg) \ \text{if} \ \dfrac{a^2 T}{2n_tD} > 1. 
        \end{cases}
    \]
    Furthermore, 
    \[||\mathcal{A}|| =
    \begin{cases}
        \Theta(1) \ \text{if}\ \dfrac{a^2 T}{2n_tD} \leq 1, \\ \\
        \Theta\bigg(\sqrt{\dfrac{a^2 T}{n_tD}}\bigg) \ \text{if}\ \dfrac{a^2 T}{2n_tD} > 1
    \end{cases}
    \]
    and $||\mathcal{A}^{-1}|| = \Theta(n_t)$.
    \label{thm: kappa}
\end{theorem}
We have included both results here for completeness but typically $\kappa_\mathcal{A} = \Theta(n_t)$ can be used since $T = n_t\Delta t$ which renders it unlikely that $a^2\Delta t/2D >1.$ We have provided the proof in \AD{Appendix} \ref{app: cond_num_A}.

\subsection{Linear equations method}
\label{ssec: lin_eqns}

In this section we describe the first classical method of solving the system of linear equations in \AD{Eq.}~\eqref{eqn: syslineareqns}. One approach for solving sparse systems of linear equations is the conjugate gradient method \cite{Shewchuk1994AnPain}. This method can solve a system of $N$ linear equations, with at most $s$ unknowns each, and a corresponding matrix $\mathcal{A}$ with condition number $\kappa_\mathcal{A}$ up to accuracy $\delta$ in the energy norm $||\cdot||_\mathcal{A}$ in time $O(s\sqrt{\kappa_\mathcal{A}}N\log(1/\delta))$. A flowchart and summary of the complexity analysis can be found in Fig. \ref{fig:classical_linear_flow}.

\begin{figure}[!ht]
\centering
\includegraphics[]{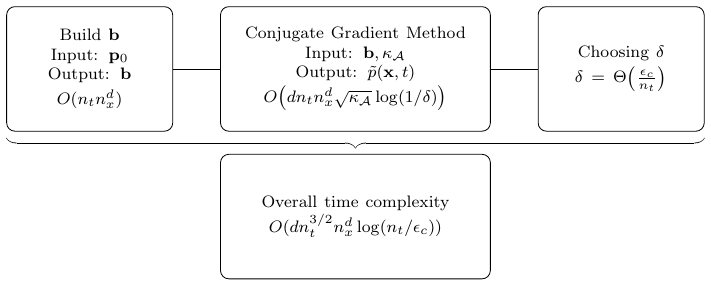}
\caption{A flowchart demonstrating the contributing factors to the overall time complexity of the classical linear equations method described by \AD{Theorem} \ref{thm: Lineqn}. The complexity for the conjugate gradient method is provided in Ref.~\cite{Shewchuk1994AnPain}.}
    \label{fig:classical_linear_flow}
\end{figure}

\begin{theorem}[Classical linear equations method] There is a classical algorithm that outputs an approximate solution $\tilde{p}(\mathbf{x}, t)$ to \AD{Eq.}~\eqref{eqn: DDE} via the system of linear equations $\mathcal{A}\tilde{\mathbf{p}} = \mathbf{b}$ from \AD{Eq.}~\eqref{eqn: syslineareqns} such that $||\tilde{p}(\mathbf{x}, t) - p(\mathbf{x}, t)||_{\infty} \leq \epsilon_c$ for all $(\mathbf{x},t) \in G$ in time
    \begin{equation}
        O\left(d n_t^{3/2}n_x^d \log\left(\frac{n_t}{\epsilon_c}\right)\right)= O\bigg(\frac{d^{4+d/2}T^{3 + d/2}}{D^{d/2}} \bigg(\frac{\zeta(aL + D)^2}{\epsilon_c }\bigg)^{(3 + d)/2}\log\bigg(\frac{T d \zeta (aL + D)^2}{\epsilon_c L^d}\bigg)\bigg).
    \end{equation}
    Suppressing the logarithmic terms this is 
    \begin{equation}
         \tilde{O}(d n_t^{3/2}n_x^d) = \tilde{O}\bigg(\frac{d^{4+d/2}T^{3 + d/2}}{D^{d/2}} \Big(\frac{\zeta(aL + D)^2}{\epsilon_c}\Big)^{(3 + d)/2}\bigg).
    \end{equation}
    \label{thm: Lineqn}
\end{theorem}
\begin{proof}
    Using Corollary \ref{cor:dxdt bounds} and \AD{Theorem} \ref{thm: kappa}, we achieve a classical \AD{discretization} accuracy $\epsilon_c$ in the $\infty$-norm with a system of $N = O(n_tn_x^d)$ linear equations, containing $O(d)$ variables and a condition number $\kappa_\mathcal{A} = \Theta(n_t)$.
    
    We can calculate $\bf{b}$ from the \AD{right-hand} side of \AD{Eq.}~\eqref{eqn: syslineareqns} in time $O(dn_x^d)$ by multiplying $\mathbf{p}_0$ and $\mathcal{L}$.
     In the following step it can be seen that this is negligible. Employing the conjugate gradient method the system of linear equations in \AD{Eq.}~\eqref{eqn: syslineareqns} can be solved with accuracy $\delta$ in the energy norm in time $O(d n_t^{3/2}n_x^d \log(1/\delta))$ \cite{Shewchuk1994AnPain}. 
    The dependence on $1/\delta$ is logarithmic so using any other norm within reason makes little difference to the complexity bound~\cite{Linden2022QuantumEquation}.
    For example,
    \begin{equation}
        ||\tilde{\mathbf{y}}-\mathbf{y}||_2=||\mathcal{A}^{-1/2}\, \mathcal{A}^{1/2}(\tilde{\mathbf{y}}-\mathbf{y})||_2 \leq ||\mathcal{A}^{-1/2}||\ ||\mathcal{A}^{1/2}(\tilde{\mathbf{y}}-\mathbf{y})||_2=||\mathcal{A}^{-1}||^{1/2}||\tilde{\mathbf{y}}-\mathbf{y}||_{\mathcal{A}}
    \end{equation}
    by triangle inequality. Using this, Theorem \ref{thm: kappa} and $\delta = \Theta(\epsilon_c/n_t)$ we achieve $\epsilon_c$ accuracy in the $2$-norm (and therefore the $\infty$-norm). This gives an overall complexity of $O(d n_t^{3/2}n_x^d \log(n_t/\epsilon_c))$. Inserting the expressions for $n_t$ and $n_x$  from Corollary \ref{cor:dxdt bounds} results in complexity
    \begin{equation}
        O\bigg(\frac{d^{4+d/2}T^{3 + d/2}}{D^{d/2}} \bigg(\frac{\zeta(aL + D)^2}{\epsilon_c }\bigg)^{(3 + d)/2}\log\bigg(\frac{T d \zeta (aL + D)^2}{\epsilon_c L^d}\bigg)\bigg)
    \end{equation}
    as stated.
    Suppressing the logarithmic terms this is 
    \begin{equation}
         \tilde{O}(d n_t^{3/2}n_x^d) = \tilde{O}\bigg(\frac{d^{4+d/2}T^{3 + d/2}}{D^{d/2}} \Big(\frac{\zeta(aL + D)^2}{\epsilon_c}\Big)^{(3 + d)/2}\bigg).
    \end{equation}
\end{proof}
This method is an application of Ref.~\cite[\AD{Theorem} 5]{Linden2022QuantumEquation} and is known to be applicable to a broader range of PDEs.
 The next method takes advantage of our forward-time discretization, and provides a significantly simpler and more efficient solution.

\subsection{Time-evolution method}
\label{ssec: time_evolution}
For problems using forward-time discretization, the simplest linear solution method is a straightforward time evolution: where we apply the operator $\mathcal{L}$ to $\mathbf{p}_0$, $n_t$ times. This approach is significantly simpler and more efficient than solving the full system of linear equations. See Fig. \ref{fig: time_flow} for a complexity analysis summary.
\begin{figure}[!ht]
\centering
\includegraphics[]{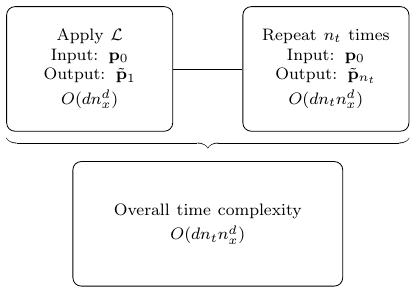}
    \caption{A flowchart demonstrating the contributing factors to the overall time complexity of the classical \AD{time-evolution} method described by \AD{Theorem} \ref{thm: timestepping}.}
    \label{fig: time_flow}
\end{figure}

\begin{theorem}[Classical time-evolution method] There is a classical algorithm that produces an approximate solution $\tilde{p}(\mathbf{x}, t)$ to \AD{Eq.}~\eqref{eqn: DDE} such that $||\tilde{p}(\mathbf{x}, t) - p(\mathbf{x}, t)||_{\infty} \leq \epsilon_c$ for all $(\mathbf{x}, t) \in G$ in time 
    \begin{equation}
        O(dn_tn_x^d) = O\bigg(\frac{d^{d/2 + 3}T^{d/2 + 2}}{D^{d/2}}\bigg(\frac{\zeta(aL + D)^2}{\epsilon_c}\bigg)^{d/2+1}\bigg).
    \end{equation}
    \label{thm: timestepping}
\end{theorem}
\begin{proof}
   We apply the operator $\mathcal{L}$ as defined in \AD{Eq.} (\ref{eqn: diff operator}) $n_t$ times to the initial vector $\mathbf{p}_0$. Each application can be performed in time $O(dn_x^d)$, \AD{and} therefore all required $\tilde{\mathbf{p}}_{i}$ can be calculated in time $O(dn_tn_x^d)$. Using the bounds from Corollary \ref{cor:dxdt bounds} for $n_t$ and $n_x$ gives the claimed result.
\end{proof}
We have shown the \AD{time-evolution} method in \AD{Theorem} \ref{thm: timestepping} is more efficient than the linear equations method in \AD{Theorem}~\ref{thm: Lineqn} by $\Tilde{O}(\sqrt{n_t})$.
However, there is no efficient quantum equivalent to \AD{Theorem} \ref{thm: timestepping} as is shown by \AD{Theorem} \ref{thm: QTM}. 
The quantum linear equations method in Sec.~\ref{ssec: q_lin_eqns} will thus be compared to the method in \AD{Theorem}~\ref{thm: Lineqn}.

\subsection{Random walk}
An alternative solution method involves using a random walk to solve the DDE in \AD{Eq.}~\eqref{eqn: DDE}. This approach exploits the stochastic nature of operator $\mathcal{L}$.
We utilise the coupling from the past (CFTP) method~\cite{Propp1998HowGraph} to obtain samples of $\tilde{p} (\mathbf{x},t)$. Starting from the initial probability distribution $\mathbf{p}_0$, we iterate the random walk $n_t$ times to generate samples from the distributions at each subsequent time step.
However, to achieve the accuracy required by \AD{Eq.}~\eqref{eqn: approx_definition} an exact sample of the probability distribution is needed (using the definition of exact sampling found in Ref. \cite{Mackay1995InformationAlgorithms}).  We demonstrate the overall complexity in the flowchart in Fig. \ref{fig: Classical_RW}.

\begin{figure}[!ht]
\centering
\includegraphics[]{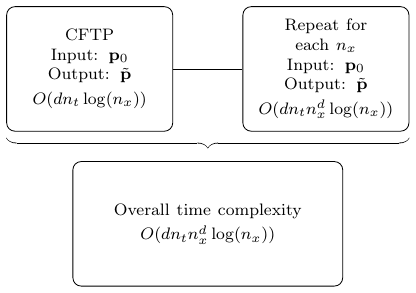}
    \caption{A flowchart demonstrating the contributing factors to the overall time complexity of the classical \AD{random-walk} method described by \AD{Theorem} \ref{thm: timestepping}. Here CFTP represents the coupled from the past method~\cite{Propp1998HowGraph}.}
    \label{fig: Classical_RW}
\end{figure}

\begin{theorem}[Random walk]
    There exists a classical algorithm that outputs an exact sample from the distribution $\tilde{\mathbf{p}}_k$ such that $||\tilde{p} (\mathbf{x},t) -p(\mathbf{x},t)||_{\infty} \leq \epsilon_c$ for all $k = 0,...,n_t $ in time 
    \begin{equation}
        O(n_tn_x^dd\log n_x) = O\bigg(\frac{d^{d/2 + 3}T^{d/2 + 2}}{D^{d/2}}\bigg(\frac{\zeta(aL + D)^2}{\epsilon_c}\bigg)^{d/2+1}\log\bigg(\frac{Td\zeta(aL + D)^2}{\epsilon_c D}\bigg)\bigg).
    \end{equation}
    Suppressing the logarithmic terms this is  
    \begin{equation}
        \tilde{O}(dn_tn_x^d) = \tilde{O}\bigg(\frac{d^{3+d/2}T^{2+d/2}\zeta^{1+d/2}}{\epsilon_c^{(d/2+1)}D^{d/2}}(aL + D)^{2+d}\bigg).
        \label{eqn: comp_random_walk}
    \end{equation}
    \label{thm: exact_sampling}
\end{theorem}
\begin{proof}
We use the CFTP method to obtain an efficient exact sampling. Since for any probability distribution we can consider a random variable $X$ of the probability distribution. 
Then this method allows us to use $\mathcal{L}$ to map the random variable $X$ to a final distribution $\tilde{\mathbf{p}}_k$. This mapping runs in time $O(H\log n_x^d)$ where $H$ is the maximum hitting time. The hitting time is defined as the maximum time required to move from any state $k$ to state $j$,
which in our case is $H = n_t$. 

Crucially, exact sampling necessitates we initiate the sampling process from each of the $n_x^d$ points at $t = 0$ in $G$. 
Consequently, the overall time complexity for generating a complete exact sample is $O(n_tn_x^dd\log n_x)$. Substituting in the values for $n_t$ and $n_x$ from Corollary \ref{cor:dxdt bounds} yields the shown result. 
\end{proof}
The method's requirement to sample from all $n_x^d$ points in space to produce an exact sampling renders it inefficient. Consequently, its $\epsilon_c$ dependence mirrors that of the time-evolution method in \AD{Theorem} \ref{thm: timestepping}. Moreover, considering the scaling of each variable, the overall efficiency in \AD{Eq.}~\eqref{eqn: comp_random_walk} exhibits worse performance than the time-evolution method in \AD{Theorem} \ref{thm: timestepping}. 
If only the expected value, not the full distribution, is needed, \AD{then} a sampling method such as a random walk would likely be more efficient.
This was demonstrated for the integral of the solution to a heat equation~\cite{Linden2022QuantumEquation} which required fewer samples than the method presented in \AD{Theorem} \ref{thm: exact_sampling}.

\subsection{\AD{Diagonalization} by Fourier transform}

While the \AD{time-evolution} method in Sec. \ref{ssec: time_evolution} exhibits the most efficient complexity scaling thus far, a further improvement is possible when we are only interested in the final time result. Diagonalizing $\mathcal{L}$ using the FFT allows for the calculation of multiple time steps simultaneously. This approach requires only knowledge of the eigenvalues, $l_j$, of $\mathcal{L}$. The method and complexity for computing these eigenvalues are detailed in the following lemma.

\begin{lemma}
    For any $\tau \in {0,...,n_t}$, and any $\delta >0$, all of the eigenvalues of $\mathcal{L}^\tau$ can be computed up to accuracy $\delta$ in time $O(dn_x^d + n_x\log(\tau/\delta))$.
    \label{lemma: eigen_L}
\end{lemma}
\begin{proof}
   From Eqns. (\ref{eqn: matrix_M}) and \eqref{eqn: M to L}, 
       \begin{equation}
        \mathcal{L} = \mathcal{I}_{n_x}^{\otimes d} + \Delta t \sum_{j = 1}^{d} \mathcal{I}_{n_x}^{\otimes (j-1)}\otimes \mathcal{H} \otimes \mathcal{I}_{n_x}^{\otimes (d-j)}.
        \label{eqn: matrixL}
    \end{equation}
    Then using Eqns. \eqref{eqn: eigen_L} and \eqref{eqn: eigen_M}, the eigenvalues of $\mathcal{L}$ are
    \begin{equation}
        l_j = 1 - \frac{4Dd\Delta t}{\Delta x^2}\sin^2\bigg(\frac{\pi j}{n_x}\bigg) + \AD{i} \frac{a d\Delta t }{\Delta x}\sin\bigg(\frac{2\pi j}{n_x}\bigg)
        \label{eqn: eigen_L_full}
    \end{equation}
    for $j\in{0,...,n_x-1}$. Using Corollary \ref{cor:dxdt bounds}, we choose $\Delta t$ and $\Delta x$ such that $\Delta t = \Delta x^2/2dD$. Then $|l_j| \in [0,1]$ 
     while $a\Delta x /2D \leq  1$ which as discussed in Sec. \ref{sec: tech_ing} is a likely bound.
    If each eigenvalue of $\mathcal{L}$ is computed up to accuracy $\delta'$ then to accomplish an accuracy of $\delta$ for the corresponding eigenvalue of $\mathcal{L}^{\tau}$, \AD{then} it is sufficient to take $\delta' = \delta/\tau$. 
    As, given the approximation $\tilde{l} = l \pm \delta'$, so $|l^{\tau} - \tilde{l}^{\tau}|\leq \tau\delta'$. Therefore, each eigenvalue only needs to be computed up to accuracy $O(\delta/\tau)$. 
    This accuracy is achieved by Taylor expansion of the trigonometric functions up to $O(\log(\tau/\delta))$ terms.
    There are $n_x^d$ eigenvalues of $\mathcal{L}$ and so the complexity of computing all of the eigenvalues is $O(n_x^d\log(\tau/\delta))$.
\end{proof}
The first step in this method is the FFT. Using the \AD{precalculated} eigenvalues from Lemma~\ref{lemma: eigen_L}, we can implement the diagonal matrix $\Lambda^{n_t}$, thereby evolving the solution through multiple time steps with a single calculation. 
This method is described in the theorem below and summarized in the flowchart in Fig. \ref{fig: FFT_flow}.

\begin{figure}[!ht]
\centering
\includegraphics[]{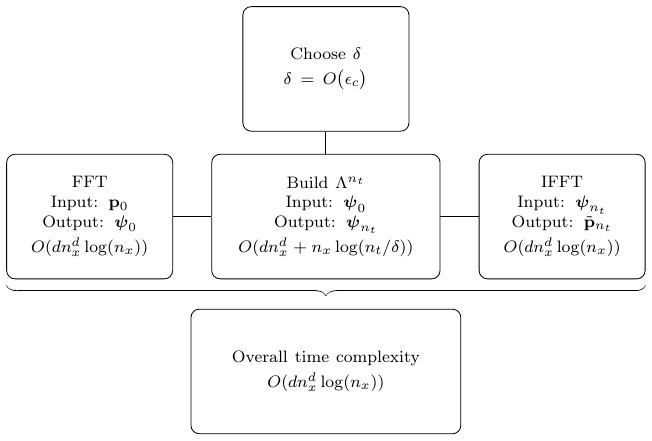}
    \caption{A flowchart demonstrating the contributing components to the overall time complexity of the classical \AD{time-stepping} method in \AD{Theorem} \ref{thm: Classical_FFT}. The complexity for the FFT (and its inverse) is as described in Ref. \cite{Frigo2005TheFFTW3}. $\Lambda$ denotes the diagonal matrix made up of the eigenvalues of $\mathcal{L}$ and the complexity is derived in Lemma \ref{lemma: eigen_L}.}
    \label{fig: FFT_flow}
\end{figure}

\begin{theorem}[Classical \AD{diagonalization} by fast Fourier transform] There is a classical algorithm that outputs an approximate solution $\tilde{p}(\mathbf{x}, t)$ such that $||\tilde{p}(\mathbf{x}, t) - p(\mathbf{x}, t)||_{\infty} \leq \epsilon_c$ for all $(\mathbf{x}, T)\in G$ at final time $T$, in time
    \begin{equation}
        O(dn_x^d \log n_x) = O\bigg(d^{d/2+1}\bigg(\frac{T\zeta L^2(aL + D)^2}{\epsilon_c D}\bigg)^{d/2} \log \bigg(\frac{Td\zeta L^2(aL + D)^2}{\epsilon_c D}\bigg)\bigg).
    \end{equation}
    For clarity, suppressing logarithmic terms this is 
    \begin{equation}
        \tilde{O}(dn_x^d) = \tilde{O}\bigg(\frac{d^{d/2+1}(T\zeta)^{d/2}L^d(aL + D)^d}{\epsilon_c^{d/2} D^{d/2}}\bigg).
    \end{equation}
    \label{thm: Classical_FFT}
\end{theorem} 
\begin{proof}
    Lemma \ref{lemma: eigen_L} demonstrated that $\mathcal{L}$ is a sum of circulant matrices operating on $d$ independent dimensions. 
   Therefore, $\mathcal{L}$ is \AD{diagonalize}d by the $d$\AD{th} tensor power of the discrete Fourier transform. The following expression is used to approximately compute $\mathbf{\tilde{p}}_i$:
    \begin{equation}
        \tilde{\mathbf{p}}_i = \mathcal{L}^i\mathbf{p}_0 = (\mathcal{F}^{\otimes d})^{-1}\Lambda^i\mathcal{F}^{\otimes d}\mathbf{p}_0\AD{,}
    \end{equation}
    \AD{w}here $\Lambda$ is the diagonal matrix whose entries are the eigenvalues of 
    $\mathcal{L}$, and $\mathcal{F}$ is the discrete Fourier transform. 

We first construct $\mathbf{p}_0$ in $O(n_x^d)$ time. Then we apply the multidimensional FFT to $\mathbf{p}_0$ to yield an intermediary vector $\boldsymbol{\psi}_0$ in $O(dn_x^d\log n_x)$ time \cite{Frigo2005TheFFTW3}. Subsequently, we multiply $\boldsymbol{\psi}_0$ by the eigenvalues of $\mathcal{L}^i$, computed to accuracy $\delta$ using Lemma~\ref{lemma: eigen_L} in time $O( n_x^d\log(\tau/\delta))$. We achieve this by building the diagonal matrix $\tilde{\Lambda}^i$ such that $||\tilde{\Lambda}^i - \Lambda^i|| \leq \delta$. Then
    \begin{equation}
        ||(\mathcal{F}^{\otimes d})^{-1}\tilde{\Lambda}^i\mathcal{F}^{\otimes d}\mathbf{p}_0 - (\mathcal{F}^{\otimes d})^{-1}\Lambda^i\mathcal{F}^{\otimes d}\mathbf{p}_0||_2 \leq ||\tilde{\Lambda}^i - \Lambda^i|| ||\mathbf{p}_0||_2 \leq \delta ||\mathbf{p}_0||_1 = \delta,
    \end{equation}
    where as in \AD{Eq.} \eqref{eqn: norm_p0}, $||\mathbf{p}_0||_1 = 1$. Therefore, it is sufficient to take $\delta = \epsilon_c$ and $\tau = n_t$. Then the total complexity of this step is
    \begin{equation}
        O\bigg(n_x^d\log \bigg(\frac{n_t}{\epsilon_c}\bigg)\bigg) = O(n_x^d(\log n_t + \log (1/\epsilon_c))).
        \label{eqn: fft_complex}
    \end{equation}
    We note that from Corollary \ref{cor:dxdt bounds} that
        \begin{equation}
            n_t = \frac{n_x^2 TdL^2}{2D},
    \end{equation} 
    so $\log n_t = O(\log n_x)$ as well as $n_x \propto 1/\sqrt{\epsilon_c}$. Then we can simplify \AD{Eq.}~\eqref{eqn: fft_complex} to $O(n_x^d(d\log n_x))$. 
    
    Therefore, we have shown that the dominant factor determining the method's overall complexity is a combination of the FFT and the eigenvalue computation. Both contribute equally to the overall time complexity.
    Inserting the values for $n_x$ from Corollary \ref{cor:dxdt bounds} results in
    \begin{equation}
        O(dn_x^d \log n_x) = O\bigg(d^{d/2+1}\bigg(\frac{T\zeta L^2(aL + D)^2}{\epsilon_c D}\bigg)^{d/2} \log \bigg(\frac{Td\zeta L^2(aL + D)^2}{\epsilon_c D}\bigg)\bigg)
    \end{equation}
    as stated.
\end{proof}

This method closely resembles that for solving the heat equation~\cite{Linden2022QuantumEquation}. 
This is to be expected, as the heat equation and the DDE are linear PDEs. %largely the same under Fourier transform (differing only by a complex shift).
The only difference in the overall complexity stems from our seeking on approximating the entire probability distribution (to within $\epsilon_c$, as detailed in \AD{Eq.}~\eqref{eqn: approx_definition}). This is unlike the heat equation~\cite{Linden2022QuantumEquation}, which seeks the integral of the distribution. Since we have shown that the method in \AD{Theorem}~\ref{thm: Classical_FFT} represents the most efficient classical approach, it will serve as the benchmark for our quantum methods.
   
\section{Quantum Method Descriptions}
\label{sec:Quantum Methods}
In this section we discuss the complexity analysis of four quantum algorithms for solving the DDE given in \AD{Eq.} (\ref{eqn: DDE}). First, it is necessary to set up the problem to be solved on a quantum computer by introducing a few technical ingredients.
These include the norm of $\mathcal{L}$ and an upper bound on the condition number of $\mathcal{L}$. We also require the condition number of $\mathcal{A}$ as in the classical case.

\subsection{Technical ingredients}
\label{ssec: Qtech_ing}
For each of the following methods we will continue to use the \AD{discretization} approximation $||\tilde{\mathbf{p}} - \mathbf{p}||_{\infty} \leq \epsilon_c $. Then the quantum state that represents $\tilde{\mathbf{p}}$ is $|\tilde{p}\rangle$ as described in Eq.~\eqref{eqn: quantum_state}.
Each of the following quantum methods will construct this state $|\tilde{p}\rangle$. \AD{Theorem} \ref{thm: q_meas} (presented later in this section) establishes the computational cost of extracting the approximated probability distribution $\tilde{\tilde{\mathbf{p}}}$ from  $|\tilde{p}\rangle$ such that $||\tilde{\tilde{\mathbf{p}}} - \tilde{\mathbf{p}}||_{\infty}\leq \epsilon_q$\AD{,} where $\epsilon_q$ represents the error in this extraction process. This cost, as defined in \AD{Theorem} \ref{thm: q_meas}, will be a multiplicative factor in the overall cost of each subsequent method. 

The quantum methods rely on a couple of further technical ingredients to allow us to more efficiently prepare the states for calculation. An upper bound on the condition number $\kappa_{\mathcal{L}}$ is provided below in Lemma \ref{lem: kappa_l}.
The condition number $\kappa_{\mathcal{L}}$ determines the time complexity for applying $\mathcal{L}$. This is used in the quantum linear equations method. 

\begin{lemma}
    The condition number of $\mathcal{L}$ from \AD{Eq.}~\eqref{eqn: M to L} is $\kappa_{\mathcal{L}} = 5$ when $\frac{D\Delta t}{\Delta x^2} \leq 1/5$ and $a/D<2\sqrt{10}$. 
    \label{lem: kappa_l}
\end{lemma}
 We use this Lemma in the preliminary part of the proof of our quantum linear equations method. The specific application criteria is deliberately restricted to guarantee a \AD{noninfinite} condition number. Since the quantum linear equations method (which is the only method that uses this lemma) will be proven to be the least efficient quantum solution method, this restriction is justified for comparative purposes.

A lower bound on the norm of $\mathcal{L}$ is required for the \AD{diagonalization} by QFT method. This is provided below in Lemma \ref{lem: bound L norm}.
\begin{lemma}[Norm of $\mathcal{L}$] 
Let $\mathcal{L}$ be defined by \AD{Eq.} (\ref{eqn: diff operator}), taking $\Delta t = \Delta x^2 /(2dD)$ as in Corollary \ref{cor:dxdt bounds}. Then for any positive integer $\tau \geq 1$,
    \begin{equation}
       \frac{1}{(4\sqrt{\tau})^d} \leq \langle 0|\mathcal{L}^{2\tau}|0\rangle = ||\mathcal{L}^{\tau}||_2^2.
    \end{equation}
    \label{lem: bound L norm}
\end{lemma}
Lemmas \ref{lem: kappa_l} and \ref{lem: bound L norm} are proved in \AD{Appendix} \ref{app: bound_l_lemma}.
With the simpler technical ingredient provided the last requirement before moving to the quantum solution methods is to define the method of quantum measurement that will be used throughout.

\subsubsection{Quantum measurement}
This section details a quantum probability distribution measurement method. 
It extends \AD{Theorem 5 in} Ref.~\cite{VanApeldoorn2021QuantumEstimation} which leverages the known quadratic speed-up of amplitude estimation in reducing complexity dependence on $\epsilon_q$~\cite{Brassard2002QuantumEstimation}. 

Our method, detailed in \AD{Theorem} \ref{thm: q_meas}, is tailored to $n_x^d$-dimensional quantum systems. 
It relies on three key lemmas. Lemma \ref{lem: prob-func} transforms a probability distribution into a corresponding function, $f(x):D^d \rightarrow [0,1]:x \rightarrow \langle x,p \rangle$ where $\langle x,p \rangle = \sum_{i = 1}^{n_x^d}p_i x_i$. 
Then Lemma \ref{lem: func-phase} uses this function to encode the probability distribution into the phase of the qubits. Finally, Lemma \ref{lem: n-dep} reduces the dependence from $n_x^d$ to $1/\epsilon_q$. \AD{Theorem} \ref{thm: q_meas} then leverages these lemmas to efficiently estimate $\tilde{\tilde{\mathbf{p}}}$ such that $||\tilde{\tilde{\mathbf{p}}} - \tilde{\mathbf{p}}||_{\infty}\leq \epsilon_q$.

\begin{lemma}[Probability algorithm to function \AD{(Lemma 3 in Ref.~{\cite{VanApeldoorn2021QuantumEstimation}})}]

 Let $U_p$ be a quantum algorithm that produces the quantum state $|\tilde{p}\rangle = \frac{1}{||\tilde{p}(\mathbf{x})||_2}\sum_{\mathbf{x}}\tilde{p}(\mathbf{x})|x\rangle$ representing a probability distribution $\tilde{p}(\mathbf{x})$ where $\mathbf{x}$ is $d$-dimensional. Let $k\geq1$ be an positive
    integer and let $D = \big\{0,\frac{1}{2^k},...,\frac{2^k-1}{2^k}\big\}$ be a \AD{discretization} of $[0,1]$. Then a quantum algorithm $U_{\tilde{f}}$ for a function $\tilde{f}$ can be constructed such that $\tilde{f}$ is an additive $\mu$-approximation of $f(x):D^d \rightarrow [0,1]:x \rightarrow \langle x,p \rangle$ using two applications of $U_p$ and $\tilde{O}(\textrm{polylog}(d/\mu))$ two-qubit gates.  
       \label{lem: prob-func}
\end{lemma}
This result relies on a quantum random access memory (QRAM) to deliver an improved gatecount compared to the case when no QRAM is used~\cite{VanApeldoorn2021QuantumEstimation}.

\begin{lemma}[Function oracle to phase oracle \AD{(Lemma 4 in Ref.~\cite{VanApeldoorn2021QuantumEstimation})}]

 Let $U_f$ be a quantum algorithm that produces a probability distribution based on the function $f:D\rightarrow [0,1]$ acting on $q$ qubits. Let $H>0$. A quantum algorithm that produces the quantum state $e^{\AD{i} Hf}(x)|x\rangle$ with $\eta$-additive error can be constructed using $O(|H| +\log(1/\eta))$ applications of $U_f$ and its inverse, and $O(q|H| + \log(1/\eta)$ two-qubit gates. 
    \label{lem: func-phase}
\end{lemma}
This result has a $q$ dependence in the gatecount, each of the methods in Sec. \ref{sec:Quantum Methods} have $q \propto n_x$. This dependence stems from Lemma \ref{lem: func-phase} considering all coordinates in $\tilde{f}$, even those with very small or $0$ entries. To achieve an ${\infty}$-norm approximation within $\epsilon_q$, we can discard coordinates with probability less than $\epsilon_q$~\cite{VanApeldoorn2021QuantumEstimation}. This reduces the number of coordinates to at most $1/\epsilon_q$, upon which we run the algorithm. Lemma~\ref{lem: n-dep} details the sampling procedure used to identify these relevant coordinates.

\begin{lemma}[\AD{(Lemma 6 in Ref.~\cite{VanApeldoorn2021QuantumEstimation})}] 

Let $p \in [0,n]$ be a probability distribution that spans $n$ components, and $\tilde{p}$ be an approximation such that $||\tilde{p} - p||_{\infty} \leq \epsilon_q$ and $\epsilon_q, \delta \in (0,1/3)$. Then $O(\log(n/\delta)/\epsilon_q)$ samples suffice to, with error probability at most $\delta$, find all $i\in [n]$ such that $\tilde{p}_i\geq \epsilon_q$. 
    \label{lem: n-dep}
\end{lemma}
\begin{proof}
    Consider a single entry $i$ such that $\tilde{p}_i\geq \epsilon_q$. After $N$ samples the probability that we have not seen $i$ yet is at most $(1-\epsilon_q)^N$. Letting $N= \frac{\log(\delta \epsilon_q)}{\log(1-\epsilon_q)} = O(\log(n/\delta)/\epsilon_q)$  ensures that this error probability is at most $\delta\epsilon_q$. 
    A union bound applied to the (at most) $1/\epsilon_q$ coordinates proves the result.
\end{proof}

In the following theorem we show how Lemmas \ref{lem: prob-func}, \ref{lem: func-phase} and \ref{lem: n-dep} can be used to extract a \AD{multidimensional} probability distribution.

\begin{theorem}[\AD{Multidimensional} probability distribution measurement]

    Let $p(\mathbf{x})$ be a probability distribution and then $\tilde{p}(\mathbf{x})$ is an $n_x^d$-component vector approximating it.
     Given an algorithm denoted as $U_{p}$ that produces a quantum state $|\tilde{p}\rangle = \frac{1}{||\tilde{\mathbf{p}}||_2}\sum_{\mathbf{x}}\tilde{p}(\mathbf{x}, t)|x\rangle$ using $q$ qubits, an approximation, $\tilde{\tilde{p}}(\mathbf{x})$, of the distribution $\tilde{p}(\mathbf{x})$ can be extracted such that $||\tilde{\tilde{p}}(\mathbf{x})-\tilde{p}(\mathbf{x})||_{\infty} \leq \epsilon_q$ with probability $1-\delta$ where $\delta$ is the error probability. This can be achieved using $O(\log(n_x^d/\delta)\log(1/(\epsilon_q\delta))/\epsilon_q)$ applications of $U_p$ and $O\Big(\frac{1}{\epsilon_q}\big(\log (1/(\epsilon_q\delta))(q + \log(1/\epsilon_q))\big)\Big)$ 2 qubit gates. This is an improved gatecount using QRAM. 
    \label{thm: q_meas}
\end{theorem}

\begin{proof}
%This method comprises several steps. 
First, we sample $|\tilde{p}\rangle$ to identify relevant variables. Next, we construct an algorithm, 
which encodes the probability distribution into the qubit phases based on the given algorithm $U_p$. This encoding uses an intermediary algorithm, $U_f$, that encodes the function $f(x):[0,1] \rightarrow \langle x, p\rangle$. 
    
    Let $k = \log(4/\epsilon_q)$. 
    \begin{enumerate}
        \item Sample directly from $U_p$ to identify the coordinates in $\mathbf{x}$ with values $\geq \epsilon_q$ as in Lemma \ref{lem: n-dep}. This takes $N = O(\log(n_x^d/\delta)/\epsilon_q)$ samples.
        	Let the number of such coordinates be $r$ which is independent of $n_x^d$. Also, $r \leq \lceil1/\epsilon_q\rceil$ 
        \item Build $r$ registers of $k$ qubits each, all in the state $\ket{0}$:
     $       |0^k\rangle...|0^k\rangle.$

        \item Apply Hadamard gates to all qubits to obtain
        \begin{equation}
            \bigotimes_{\AD{j}=1}^{r} \bigg(\frac{1}{\sqrt{2^k}}\sum_{x_i=0}^{2^k-1}|x_{\AD{j}}\rangle\bigg) = \frac{1}{2^{kr/2}}\sum_{x\in\{0,2^k-1\}^r}|x\rangle.
        \end{equation}
        \item Make a phase query for a $1/6$\AD{th} approximation of $f(x)=\langle x,p \rangle$ using Lemma \ref{lem: prob-func} with $\mu <1/(96\epsilon_q)$ and Lemma \ref{lem: func-phase} with $H=2^k$ and $\eta \leq 1/12$ which provides a state $1/6$-close in $2$-norm to 
        \begin{equation}
            \frac{1}{2^{kr/2}}\sum_{x\in\{0,2^k-1\}^r}e^{\AD{i} \langle x,p\rangle}|x\rangle = \frac{1}{2^{kr/2}}\sum_{x\in\{0,2^k-1\}^r}e^{\AD{i} \sum_{\AD{j}} x_{\AD{j}} p_{\AD{j}}}|x\rangle
        \end{equation}
        \begin{equation}
            = \frac{1}{2^{kr/2}}\sum_{x\in\{0,2^k-1\}^r} \bigg( \prod_{\AD{j}=1}^r e^{\AD{i} x_{\AD{j}} p_{\AD{j}}}\bigg)|x\rangle
        \end{equation}
        \begin{equation}
            = \bigotimes_{\AD{j}=1}^r \Bigg(\frac{1}{\sqrt{2^{k}}}\sum_{x_{\AD{j}}=0}^{2^k-1} e^{\AD{i} x_{\AD{j}} p_{\AD{j}}}|x_{\AD{j}}\rangle\Bigg).
            %bigg( \prod_{i=1}^r e^{\AD{i} x_i p_i}\bigg)
        \end{equation}
        This can be completed with $O(1/\epsilon_q)$ applications of $U_p$ and $O\big(\frac{1}{\epsilon_q}(\log(1/\epsilon_q) +q)\big)$ gates.
        \item Apply the $k$-qubit inverse QFT to each of the registers and measure them. This requires $O(1/\epsilon_q\log^2(1/\epsilon_q))$ gates.
    \end{enumerate}
  If we ignore $2$-norm error due to the imperfect phase oracle then it follows from the analysis of phase estimation that the final vector is $\tilde{p}$ such that $|\tilde{\tilde{p}}_{\AD{j}} - \tilde{p}_{\AD{j}}|\leq 4/2^k\leq \epsilon_q$ with probability at least 5/6 per coordinate. Since at most a $1/6-2$-norm error was incurred it can be said that $|\tilde{\tilde{p}}_{\AD{j}} - \tilde{p}_{\AD{j}}|\leq \epsilon_q$ with probability at least 2/3 per coordinate. By repeating the whole process $O(\log(1/\delta\epsilon_q))$ and taking the median, the error probability can be reduced to $\delta \epsilon_q$. Taking the union bound achieves the stated result.
\end{proof}
We apply this measurement protocol to all subsequent quantum solution methods. This improves the dependence on the error parameter $\epsilon_q$, resulting in a bound of $1/\epsilon_q$, an improvement over the $1/\epsilon_q^2$ bound obtained from a naive sampling approach using Chebyshev's inequality.

\subsection{Quantum linear systems}
\label{ssec: q_lin_eqns}
In this section we introduce the first quantum method  for solving systems of linear equations based on the method introduced in Ref. \cite{Chakraborty2018TheSimulation}. The computational complexity for solving systems of the form $\mathcal{A}\mathbf{x} = \mathbf{b}$ \AD{[as in \AD{Eq.} (\ref{eqn: syslineareqns})]} depends primarily on two factors: the condition number of matrix $\mathcal{A}$, $\kappa_\mathcal{A}$ and the complexity of preparing the state 
$\ket{b}$ representing vector $\mathbf{b}$ \cite{Dalzell2023QuantumComplexities}. \AD{Theorem} \ref{thm: kappa} shows that $\kappa_\mathcal{A}$ scales as $n_t$, which, while not negligible, remains finite and controllable. However, preparing the state $\ket{b}$ presents a significant challenge.

In our case, \AD{Eq.}~\eqref{eqn: syslineareqns}, only the first element of $\ket{b}$ is \AD{nonzero}, with a value of $\mathcal{L}\mathbf{p}_0.$  
Classical computation of this requires $(dn_x^d)$, as was used in \AD{Theorem} \ref{thm: Lineqn}. Alternatively, $\mathcal{L}\mathbf{p}_0$ can be prepared directly in the quantum space, with a complexity determined by the condition number of $\mathcal{L}$, $\kappa_{\mathcal{L}}$. 
We have shown in Lemma~\ref{lem: kappa_l} that $\kappa_{\mathcal{L}} = O(1)$ for a reduced problem space. 
We will consider ust this reduced problem space for this solution method.
We begin by introducing the subroutine used to construct the state $|\Tilde{\mathbf{p}}\rangle$~\cite{Chakraborty2018TheSimulation}.

\begin{theorem}[Solving Linear Equations \AD{(Theorem 10 in Ref.~\cite{Chakraborty2018TheSimulation})}] Let $\mathcal{A}\mathbf{y} = \mathbf{b}$ for an $N \times N$ invertible matrix $\mathcal{A}$ with sparsity $s$ and condition number $\kappa_\mathcal{A}$. Given an algorithm that constructs the state $|b\rangle = \frac{1}{||\mathbf{b}||_2}\sum_i \mathbf{b}_i|i\rangle$ in time $T_b$, there is a quantum algorithm that can output a state $|\Tilde{y}\rangle$ such that
    \begin{equation}
        \bigg|\bigg||\Tilde{y}\rangle - |\mathcal{A}^{-1}b\rangle \bigg|\bigg|_2 \leq \eta
    \end{equation}
    with probability at least 0.99, in time
    \begin{equation}
        O\bigg(\kappa_\mathcal{A}\bigg(s(T_U+\log N)\log^2\bigg(\frac{\kappa_\mathcal{A}}{\eta}\bigg) + T_b\bigg) \log \kappa_\mathcal{A}\bigg),
        \label{thm: q_lin_eqns}
    \end{equation}
    where, 
    \begin{equation}
        T_U = O\bigg(\log N + \log^{2.5}\bigg(\frac{s\kappa_\mathcal{A} \log(\kappa_\mathcal{A}/\eta)}{\eta}\bigg)\bigg).
    \end{equation}
    \label{thm: Solv LE}
\end{theorem}

In \AD{Theorem} \ref{thm: Solv LE} we present an algorithm for solving the DDE. We implement this algorithm in \AD{Theorem} \ref{thm: QLEM}, and summarise its complexity in Fig. \ref{fig: quantum_linear}.

\begin{figure}[!ht]
\centering
\includegraphics[]{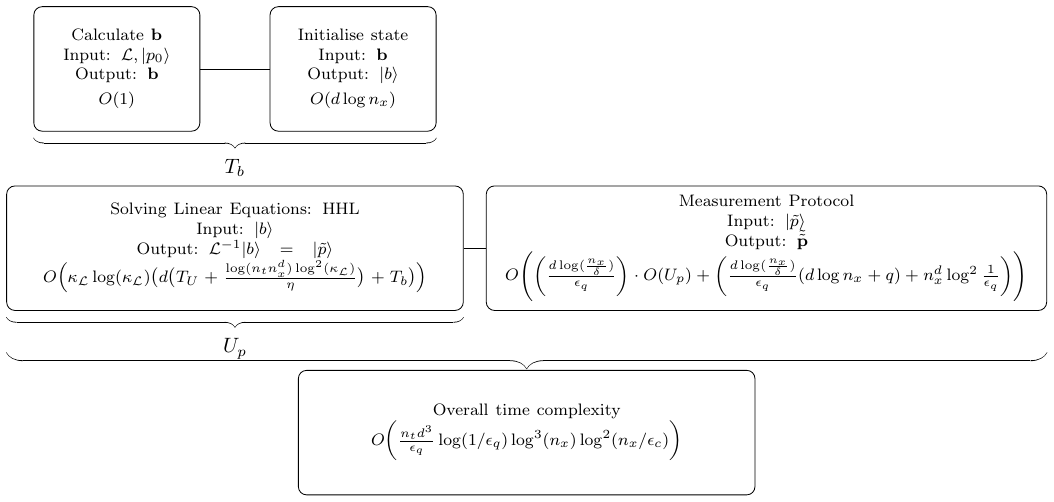}
\caption{A flowchart demonstrating the contributing factors to the overall time complexity of the quantum linear equations solution method for the DDE in \AD{Theorem} \ref{thm: QLEM}. The \AD{``solving linear equations''} complexity is based on \AD{Theorem} \ref{thm: Solv LE} and the complexity of the \AD{``measurement protocol''} is from \AD{Theorem} \ref{thm: q_meas}.}
    \label{fig: quantum_linear}
\end{figure}

\begin{theorem}[Quantum linear equations method]
  There is a quantum algorithm that approximates $p(\mathbf{x},t)$ in the case when $d \geq 3$ so $\frac{D\Delta t}{\Delta x^2} \leq 1/5$ by \AD{Theorem} \ref{thm: approx linf} and $a/D<2\sqrt{10}$. The approximation achieves $\Tilde{\Tilde{p}}(\mathbf{x}, t)$ such that $||\Tilde{\Tilde{p}}(\mathbf{x}, t) - p(\mathbf{x}, t)||_{\infty} \leq \epsilon_c +\epsilon_q $ for all $(\mathbf{x},t) \in G$ with $99\% $ success probability in time
  \begin{multline}
        O\bigg(\frac{n_td^3}{\epsilon_q}\log(1/\epsilon_q)\log^3(n_x)\log^2(n_x/\epsilon_c)\bigg) = \\ 
       O\bigg(\frac{d^5 T^2 \zeta (aL + D)^2}{\epsilon_q\epsilon_c}\log(1/\epsilon_q)\log^3\Big(\frac{Td\zeta L^2 (aL + D)^2}{\epsilon_c D}\Big) \log^2\Big(\frac{Td\zeta L^2(aL + D)^2}{\epsilon_c^3 D}\Big)\bigg).
    \end{multline} 
    Suppressing the logarithmic terms, this is
    \begin{equation}
            \Tilde{O}\bigg(\frac{d^5T^2 \zeta (aL + D)^2}{\epsilon_q\epsilon_c} \bigg)\AD{.}
    \end{equation}
    \label{thm: QLEM}
\end{theorem}

\begin{proof}
We use Corollary \ref{cor:dxdt bounds} and \AD{Theorem} \ref{thm: kappa}, to achieve a classical \AD{discretization} accuracy of $\epsilon_c$ in the $\infty$-norm by a system of $N = O(n_tn_x^d)$ linear equations. The system has a condition number $\kappa_\mathcal{A} = \Tilde{O}(n_t)$ as defined in \AD{Theorem} \ref{thm: kappa} and sparsity $O(d)$. Then \AD{Theorem} \ref{thm: Solv LE} can be used to solve the system of linear equations.

 First, we construct the initial quantum state $|b\rangle$ on the right-hand side of \AD{Eq.} (\ref{eqn: syslineareqns}). This provides $T_b$ for \AD{Theorem} \ref{thm: Solv LE}. We begin by constructing the state $|p_0\rangle$ and assume that the marginals of $\mathbf{p}_0$ (and its powers) can be computed efficiently,  allowing this construction to be completed in time $O(d\log n_x)$. 
 Next, we apply $\mathcal{L}$ to this state, using the linear combination of unitaries \cite{Childs2015QuantumPrecision}. This step takes $\Tilde{O}(\kappa_{\mathcal{L}})$ time. Lemma~\ref{lem: kappa_l} shows that in the restricted regime we have specified $\kappa_{\mathcal{L}} = O(1)$.
Therefore, the \AD{right-hand} side of the system of linear equations can be built in time $O(d\log n_x)$, this acts as $T_b$ from \AD{Theorem} \ref{thm: Solv LE} and will be shown to be negligible. 

Using \AD{Theorem} \ref{thm: Solv LE}, there is a quantum algorithm that can produce a state $|\Tilde{p}\rangle$ such that $|||\Tilde{p}\rangle - \mathcal{A}^{-1}|\mathcal{L} p_0\rangle||_2 \leq \epsilon_c$  in time 
    \begin{equation}
        O\bigg(n_t\log n_t\Big(d\log n_x + d\log^2 \Big(\frac{n_t}{\epsilon_c}\Big)\Big(\log(n_tn_x^d) + \log^{2.5}\Big(\frac{n_td\log(n_t/\epsilon_c)}{\epsilon_c}\Big)\Big)\Big)\bigg)\AD{,} 
    \end{equation}
    \begin{equation}
        = O\big(n_td\log n_t\log^2(n_t/\epsilon_c)\log(n_tn_x^d)\big)\AD{,} 
    \end{equation}
    since 
    $\log n_tn_x^d \gg \log^{2.5}\Big(\frac{n_td\log(n_t/\epsilon_c)}{\epsilon_c}\Big)$.
     
    Then we use \AD{Theorem} \ref{thm: q_meas} to extract the probability distribution $\Tilde{\Tilde{\mathbf{p}}}$ from the state $|\Tilde{p}\rangle$, with error probability $\delta = 0.01$. This requires repeating the steps above $O(\log(n_x^d/\delta) \log(1/(\epsilon_q\delta))/\epsilon_q)$\AD{, w}hich uses an additional $O\big(1/\epsilon_q\big(\log (1/(\epsilon_q\delta))(q + \log(1/\epsilon_q))\big)\big)$ gates. In this case $q = O(\log n_tn_x^d)$. Therefore, we have shown that the overall complexity of the final distribution is
    \begin{equation}
        O\bigg(\frac{n_td}{\epsilon_q}\log(n_x^d)\log(1/\epsilon_q)\log n_t\log^2(n_t/\epsilon_c)\log(n_tn_x^d)+ 1/(\epsilon_q)\log(n_tn_x^d/\epsilon_q)\bigg) 
    \end{equation}
   \AD{and} it is clear that the gate complexity is negligible. Furthermore $\log n_t = O(\log n_x)$ so this can be reduced to
    \begin{equation}
         O\bigg(\frac{n_td^3}{\epsilon_q}\log(1/\epsilon_q)\log^3(n_x)\log^2(n_x/\epsilon_c)\bigg).
        \label{eqn: comp_lin_eqns}
    \end{equation}
    Then we substitute in the values for $n_t$ and $n_x$ from Corollary \ref{cor:dxdt bounds} resulting in an overall complexity of 
    \begin{equation}
            O\bigg(\frac{d^5 T^2 \zeta (aL + D)^2}{\epsilon_q \epsilon_c}\log(1/\epsilon_q)\log^3\Big(\frac{Td\zeta L^2 (aL + D)^2}{\epsilon_c D}\Big) \log^2\Big(\frac{Td\zeta L^2(aL + D)^2}{\epsilon_c^3 D}\Big)\bigg). 
    \end{equation} 
    For clarity, this is
    \begin{equation}
            \Tilde{O}\bigg(\frac{d^5T^2 \zeta (aL + D)^2}{\epsilon_q \epsilon_c} \bigg) 
    \end{equation}
    when we suppress logarithmic terms. 

    We consider the spatial complexity of this method during the application of \AD{Theorem} \ref{thm: q_meas}. In the first step we use $q = \log n_t n_x^d$ qubits, while the subsequent steps use $1/\epsilon_q(\AD{\log}(1/\epsilon_q)+q)$ qubits. Therefore, the overall spatial complexity is $O\bigg((1+1/\epsilon_q)\Big(d\log\Big(\frac{Td\zeta L^2 (aL+D)^2}{\epsilon_c D}\Big)\Big) + 1/\epsilon_q \log(1/\epsilon_q)\bigg) = \Tilde{O}(d/\epsilon_q)$.
\end{proof}
This is the only quantum method that provides a solution for all time steps, albeit with a predictable efficiency trade-off.
On inspection of the results of \AD{Theorem}~\ref{thm: QLEM} and \AD{Theorem}~\ref{thm: timestepping} there is a quantum advantage in time efficiency when using the quantum method. 
In \AD{Theorem} \ref{thm: QLEM} we have demonstrated that, for all controllable variables, quantum scaling surpasses the classical scaling in \AD{Theorem} \ref{thm: timestepping} provided we bound the variable $\epsilon_q$. Therefore, quantum advantage is achieved when 
\begin{equation*}
    1/\epsilon_q \lesssim \Tilde{O}\Bigg(\frac{d^{d/2-2}T{d/2}}{L^{d^2/2-d}}\Big(\frac{\zeta(aL +D)^2}{\epsilon_c}\Big)^{d/2}\Bigg)\AD{,}
\end{equation*}
while maintaining the overall error $\epsilon$. The $\epsilon_c$ dependence of the quantum linear equations method is \AD{well known}, and our results match those found in the literature \cite{Berry2014High-orderEquations, Jin2022TimeEquations, Montanaro2015QuantumMethods}. We have included this result for comparison with the following methods in Sec. \ref{ssec:qrw} and Sec. \ref{ssec: QFT} will demonstrate an improved $\epsilon$ dependence.

\subsection{Quantum time evolution}
\label{ssec: qtm}
We present here the quantum \AD{time-evolution} method described in Ref.~\cite{Over2025QuantumOperator} as a quantum equivalent to the classical method presented in \AD{Theorem}~\ref{thm: timestepping}. To evolve a solution through time with repeated applications of the operator $\mathcal{L}$ directly would not provide any computational advantage. An alternative is to evolve the solution using Hamiltonian simulation. For the DDE this can be achieved using a linear combination of unitaries to combine the Hamiltonian simulation of the drift term with a shift operator for the diffusion term~\cite{Brearley2024QuantumSimulation, Over2025QuantumOperator}. We stated at the end of Sec.~\ref{ssec: time_evolution} that no efficient quantum method exists. We will prove that statement in \AD{Theorem} \ref{thm: QTM} and summarise the complexity in Fig. \ref{fig: QTM}.

\begin{figure}
\centering
\includegraphics[]{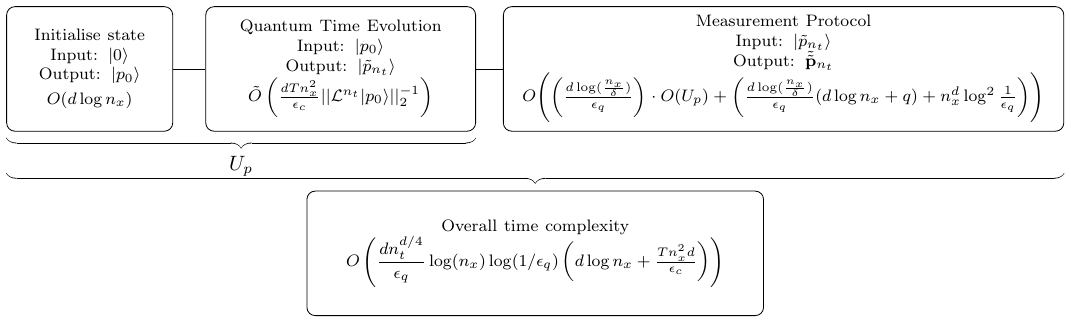}
\caption{The figure shows the complexities for the quantum \AD{time-evolution} algorithm from \AD{Theorem} \ref{thm: QTM}. The complexity of the \AD{``quantum time evolution''} is from Ref.~\cite{Over2025QuantumOperator} and the complexity of the \AD{``measurement protocol''} is from \AD{Theorem} \ref{thm: q_meas}.}
    \label{fig: QTM}
\end{figure}

\begin{theorem}
     There is a quantum algorithm that estimates $p(\mathbf{x})$ such that $||\Tilde{\Tilde{p}}(\mathbf{x}, T) - p(\mathbf{x}, T)||_{\infty} \leq \epsilon_c +\epsilon_q $ for all $(\mathbf{x}) \in G$ at a fixed $T=n_t\Delta t$ with $99\% $ success probability in time
    \begin{multline}
        O\left( \dfrac{dn_t^{d/4}}{\epsilon_q}\log(n_x) \log(1/\epsilon_q) \left(d\log n_x + \frac{Tn_x^2 d}{\epsilon_c}\right)\right)\\
        = O\left( \dfrac{d^{d/2+1} T^{d/2}\zeta^{d/4}(aL+D)^{d/2}}{\epsilon_q\epsilon_c^{d/4}}\log\left(\frac{Td\zeta L^2(aL+D)^2}{\epsilon_c D}\right) \log(1/\epsilon_q)\right.\\
        \left.
        \left(d\log \left(\frac{Td\zeta L^2(aL+D)^2}{\epsilon_c D}\right) + \frac{d^2T^2\zeta L^2(aL+D)^2}{\epsilon_c^2 D} \right) \right).
    \end{multline}
    For clarity, when logarithmic terms are suppressed this is
    \begin{equation}
        \Tilde{O}\left( \dfrac{d^{d/2+3} T^{d/2+2}\zeta^{d/4+1}L^2(aL+D)^{d/2}}{\epsilon_q\epsilon_c^{d/4+2}}\right).
    \end{equation}

    \label{thm: QTM}
\end{theorem}
\begin{proof}
    We start by producing the state $\ket{p_0}$, which can be completed in time $O(d\log n_x)$ as shown in \AD{Theorem} \ref{thm: QLEM}. We can then compute the evolution of the solution using the quantum algorithm proposed in Ref.~\cite{Over2025QuantumOperator}. This breaks down the matrix $\mathcal{L}$ into an advection-like component and a shift operator that implements the diffusion. The advection component is represented as a Hamiltonian simulation~\cite{Brearley2024QuantumSimulation} and is then linearly combined with the corrective shift operator using a linear combination of unitaries. 
    This process is repeated $n_t$ times, which can be completed in $\Tilde{O}(Tn_x^2 d/\epsilon_c)$~\cite{Over2025QuantumOperator}. The probability of success of this method is $||\mathcal{L}^{n_t}|p_0\rangle||_2^{-1}$. We know that $||\mathcal{L}^{n_t}|p_0\rangle||_2^{-1} \leq (4\sqrt{n_t})^{d/2}$ using Lemma \ref{lem: bound L norm}. Therefore, the complexity of producing the state $\ket{\Tilde{p}_{n_t}}$ is 
    \begin{equation}
        O\left(\left(d\log n_x + \frac{Tn_x^2 d}{\epsilon_c}\right)(4\sqrt{n_t})^{d/2} \right).
    \end{equation}

    Then we use \AD{Theorem} \ref{thm: q_meas} to extract the probability distribution $\Tilde{\Tilde{\mathbf{p}}}$ from the state $|\Tilde{p}\rangle$, with error probability $\delta = 0.01$. This requires repeating the steps above $O(\log(n_x^d/\delta) \log(1/(\epsilon_q\delta))/\epsilon_q)$ times, and uses an additional $O\big(1/\epsilon_q\big(\log (1/(\epsilon_q\delta))(q + \log(1/\epsilon_q))\big)\big)$ gates. In this case $q = O(d\log n_x + \log d)$~\cite{Over2025QuantumOperator}. Therefore, we have shown that the overall complexity of the final distribution is
    \begin{equation}
        O\left( \dfrac{\log(n_x^d) \log(1/\epsilon_q)}{\epsilon_q} \left(d\log n_x + \frac{Tn_x^2 d}{\epsilon_c}\right)(4\sqrt{n_t})^{d/2} +  \left(\dfrac{\log (1/\epsilon_q)(q + \log(1/\epsilon_q))}{\epsilon_q}\right)\right)
    \end{equation}
   \AD{and} it is clear that the gate complexity is negligible. Therefore, we reduce the complexity to
   \begin{equation}
        O\left( \dfrac{dn_t^{d/4}}{\epsilon_q}\log(n_x) \log(1/\epsilon_q) \left(d\log n_x + \frac{Tn_x^2 d}{\epsilon_c}\right)\right).
    \end{equation}

   Then we substitute in the values for $n_t$ and $n_x$ from Corollary \ref{cor:dxdt bounds} resulting in an overall complexity of
   \begin{multline}
        O\left( \dfrac{d^{d/2+1} T^{d/2}\zeta^{d/4}(aL+D)^{d/2}}{\epsilon_q\epsilon_c^{d/4}}\log\left(\frac{Td\zeta L^2(aL+D)^2}{\epsilon_c D}\right) \log(1/\epsilon_q)\right.\\
        \left.\left(d\log \left(\frac{Td\zeta L^2(aL+D)^2}{\epsilon_c D}\right) + \frac{d^2T^2\zeta L^2(aL+D)^2}{\epsilon_c^2 D}\right)\right).
    \end{multline}
    Or when we suppress logarithmic terms
    \begin{equation}
        \Tilde{O}\left( \dfrac{d^{d/2+3} T^{d/2+2}\zeta^{d/4+1}L^2(aL+D)^{d/2}}{\epsilon_q\epsilon_c^{d/4+2}}\right).
    \end{equation}
\end{proof}

In this method we evolve the solution with time using subsequent applications of a shifted Hamiltonian simulation algorithm achieved by a linear combination of unitaries~\cite{Over2025QuantumOperator}. We present this method as a quantum counterpart to the classical time stepping method presented in Sec.~\ref{ssec: time_evolution}. We can see that the time complexity of this method is very similar to the classical method. There is some improvement on scaling with $\zeta$, however, the scaling with $d$ and $T$ is the same and the scaling with $\epsilon_c$ and $L$ are worse. Consequently, when the new term $\epsilon_q$ is included, the quantum method offers no advantage. There is also a noisy intermediate scale quantum algorithm proposed by Ref.~\cite{Tiwari2025AlgorithmicMethod} to solve PDEs by time stepping using the quantum lattice Boltzmann method. This provides a solution at all time steps by measuring and reloading the state at each step. While interesting, it is not comparable within our framework since our framework has no corresponding attribute for the collision operator required by this method.

\subsection{Quantum random walk}
\label{ssec:qrw}
Similar to the classical approach, we present a more efficient method for directly generating the quantum state corresponding to the distribution of a random walk at time $T = n_t\Delta t$. This method employs a quantum version of the \AD{random-walk} algorithm, resulting in the state $|\Tilde{p}_{n_t}\rangle$.

\begin{theorem}[\AD{(Theorem 1 in Ref.~\cite{Apers2018QuantumTesting})}]

 Given a symmetric Markov chain with transition matrix $\mathcal{L}$ and a quantum state $|\psi_0\rangle$, there is an algorithm which produces an approximate state $|\Tilde{\psi}_k\rangle$ such that
    \begin{equation}
        \bigg|\bigg| |\Tilde{\psi}_k \rangle - \frac{\mathcal{L}^k|\psi_0\rangle}{||\mathcal{L}^k|\psi_0\rangle||_2} \bigg|\bigg|_2 \leq \eta
    \end{equation}
    using 
    \begin{equation}
        O\bigg(||\mathcal{L}^k|\psi_0\rangle||_2^{-1}\sqrt{k\log(1/(\eta ||\mathcal{L}^k|\psi_0\rangle||_2))}\bigg)
    \end{equation}
    steps of the quantum walk corresponding to $\mathcal{L}$.
    \label{thm: Quantum_walk}
\end{theorem}
We combine this with the measurement technique described in \AD{Theorem} \ref{thm: q_meas} to construct a random-walk algorithm for solving the DDE. We have provided a summary of the complexity in Fig. \ref{fig: QRW}.

\begin{figure}
\centering
\includegraphics[]{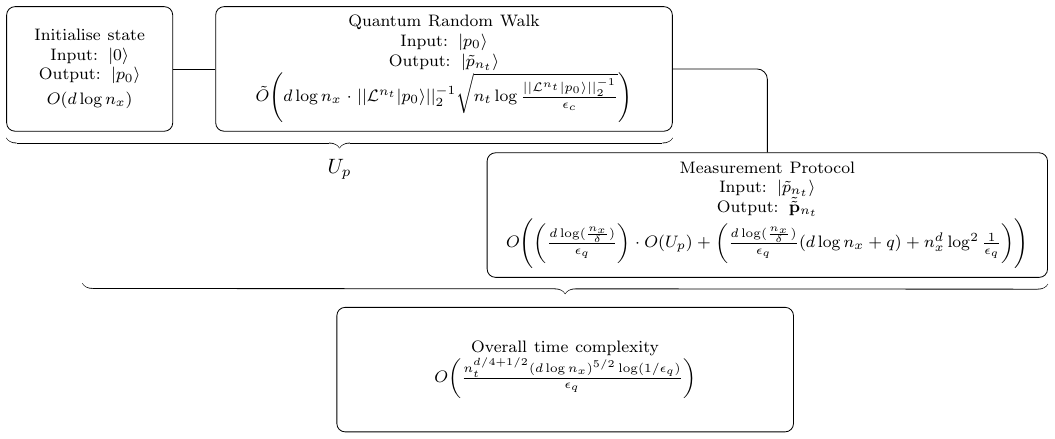}
\caption{The figure shows the complexities for the quantum \AD{random-walk} algorithm from \AD{Theorem} \ref{thm: QRW}. Where the complexity of the \AD{``quantum random walk''} is described in \AD{Theorem}~\ref{thm: Quantum_walk} and the complexity of the \AD{``measurement protocol''} is from \AD{Theorem} \ref{thm: q_meas}.}
    \label{fig: QRW}
\end{figure}

\begin{theorem}[Quantum random walk]
 There is a quantum algorithm that estimates $p(\mathbf{x})$ such that $||\Tilde{\Tilde{p}}(\mathbf{x}, T) - p(\mathbf{x}, T)||_{\infty} \leq \epsilon_c +\epsilon_q $ for all $(\mathbf{x}) \in G$ at a fixed $T=n_t\Delta t$ with $99\% $ success probability in time
    \begin{multline}
    O\bigg(\frac{n_t^{d/4+1/2}(d\log n_x)^{5/2}\log(1/\epsilon_q)}{\epsilon_q}\bigg)\\
    = O\bigg(\frac{d^{(d+7)/2}}{\epsilon_q}\Big(\frac{T^2 \zeta (aL + D)^2}{\epsilon_c }\Big)^{d/4+1/2} \log^{5/2}\Big(\frac{Td\zeta L^2(aL + D)^2}{\epsilon_c D}\Big) \log(1/\epsilon_q)\bigg).
     \end{multline}
    For clarity, when logarithmic terms are suppressed this is
    \begin{equation}
         \Tilde{O}\bigg(\frac{d^{5/2}n_t^{d/4+1/2}}{\epsilon_q}\bigg)=        \Tilde{O}\bigg(\frac{d^{(d+7)/2}T^{d/2+1}\zeta^{d/4+1/2} (aL + D)^{d/2+1}}{\epsilon_q\epsilon_c^{d/4+1/2}} \bigg).
    \end{equation}
    \label{thm: QRW}
\end{theorem}

\begin{proof}
     We start by producing the initial state $|p_0\rangle$. Given we have assumed that sums of squares of $p_0$ over arbitrary regions can be computed in time $O(1)$, $|p_0\rangle$ can be computed in time $O(d\log n_x)$ using the techniques from \AD{Refs.}~\cite{Zalka1998SimulatingComputer, Sanders2018Black-boxArithmetic}. Next we use the algorithm from \AD{Theorem} \ref{thm: Quantum_walk} to produce the state $|\Tilde{p}_{n_t}\rangle$. The complexity of implementing a quantum walk step is $O(d\log n_x)$~\cite{Linden2022QuantumEquation}. 
    
     This results in a complexity of building the state $|\Tilde{p}_{n_t}\rangle $ as
     \begin{equation}
         O\bigg(d \log n_x \Big(1+ ||\mathcal{L}^{n_t}|p_0\rangle||_2^{-1} \sqrt{n_t\log \Big(\frac{||\mathcal{L}^{n_t}|p_0\rangle||_2^{-1}}{\epsilon_c}\Big)}\Big)\bigg).
     \end{equation}
     We know that $||\mathcal{L}^{n_t}|p_0\rangle||_2^{-1} \leq (4\sqrt{n_t})^{d/2}$ using Lemma \ref{lem: bound L norm}, so
     \begin{equation}
         O\bigg(d \log n_x \Big(1+ (\sqrt{n_t})^{d/2} \sqrt{n_t\log \Big(\frac{(\sqrt{n_t})^{d/2}}{\epsilon_c}\Big)}\Big)\bigg)
     \end{equation}
     \begin{equation}
         = O\bigg(d n_t^{d/4} \log n_x  \sqrt{n_t(\log(1/\epsilon_c)+ d/4\log n_t)}\bigg).
     \end{equation}
     To further simplify the expression, we apply the following limits; $\log n_t = O(\log n_x)$ and $n_t \propto 1/\epsilon_c$ which follow from Corollary \ref{cor:dxdt bounds} (as shown in the proof of \AD{Theorem}~\ref{thm: Classical_FFT}).  
     So we can simplify the overall complexity of producing the state $|\Tilde{p}_{n_t}\rangle $ to 
     \begin{equation}
        O\bigg(d^{3/2} n_t^{d/4+1/2}(\log n_x)^{3/2}\bigg).
     \end{equation}

    Then we use \AD{Theorem} \ref{thm: q_meas} to extract the probability distribution $\Tilde{\Tilde{\mathbf{p}}}_{n_t}$ from the state $|\Tilde{p}_{n_t}\rangle$, such that $||\Tilde{\Tilde{\mathbf{p}}}_{n_t}-\Tilde{\mathbf{p}}_{n_t}||_{\infty}\leq \epsilon_q $, which requires repeating the steps above $O(\log(n_x^d/\delta) \log(1/(\epsilon_q\delta))/\epsilon_q)$.  This requires an additional $O\big(1/\epsilon_q\big(\log (1/(\epsilon_q\delta))(q + \log(1/\epsilon_q))\big)\big)$ gates. 
    This results in an overall complexity for the final distribution $\Tilde{\Tilde{\mathbf{p}}}_{n_t} $ such that $||\Tilde{\Tilde{\mathbf{p}}}_{n_t}-\mathbf{p}_{n_t}||_{\infty}\leq \epsilon_c +\epsilon_q $ is
    \begin{equation}
         O\bigg(\frac{\log(n_x^d/\delta) \log(1/(\epsilon_q\delta))}{\epsilon_q}\Big(d^{3/2} n_t^{d/4+1/2}(\log n_x)^{3/2}\Big) + \frac{1}{\epsilon_q}\big(\log (1/(\epsilon_q\delta))(q + \log(1/\epsilon_q))\big)\bigg).
     \end{equation}
      In this case $q = d\log n_x$ and $\delta = 0.01$ 
     \begin{equation}
         O\bigg(\frac{n_t^{d/4+1/2}(d\log n_x)^{5/2}\log(1/\epsilon_q)}{\epsilon_q} + \frac{\log(1/\epsilon_q)}{\epsilon_q}(d\log n_x + \log(1/\epsilon_q))\bigg)
     \end{equation}
     \begin{equation}
         = O\bigg(\frac{n_t^{d/4+1/2}(d\log n_x)^{5/2}\log(1/\epsilon_q)}{\epsilon_q}\bigg).
     \end{equation}
     Then we apply the values of $n_x$ and $n_t$ from Corollary~\ref{cor:dxdt bounds} such that the overall complexity can be written as
     \begin{equation}
        O\bigg(\frac{d^{(d+7)/2}}{\epsilon_q}\Big(\frac{T^2 \zeta (aL + D)^2}{\epsilon_c }\Big)^{d/4+1/2} \log^{5/2}\Big(\frac{Td\zeta L^2(aL + D)^2}{\epsilon_c D}\Big) \log(1/\epsilon_q)\bigg).
     \end{equation}
    For clarity, when we suppress logarithmic terms this is
    \begin{equation}
         \Tilde{O}\bigg(\frac{d^{5/2}n_t^{d/4+1/2}}{\epsilon_q}\bigg)=        \Tilde{O}\bigg(\frac{d^{(d+7)/2}T^{d/2+1}\zeta^{d/4+1/2} (aL + D)^{d/2+1}}{\epsilon_q\epsilon_c^{d/4+1/2}} \bigg).    
     \end{equation}

     As in the previous case, we consider the spatial complexity during the measurement step in \AD{Theorem} \ref{thm: q_meas}. Since $\log n_t = O(\log n_x)$ this result is largely the same as for the previous case $O\big((1+1/\epsilon_q)\Big(d\log\Big(\frac{Td\zeta L^2(aL+D)^2}{\epsilon_c D}\Big)\Big) + 1/\epsilon_q \log(1/\epsilon_q)\big) = \Tilde{O}(d/\epsilon_q)$.
\end{proof}
In this method we leverage the inherent stochasticity of the DDE for the values of $\Delta t$ and $\Delta x$ we chose in Corollary \ref{cor:dxdt bounds}.  Consequently, the applicability of this method extends to a broader class of stochastic PDEs, including \AD{nonlinear} PDEs with spatially varying parameters. This is unsurprising given the well-established stochastic nature of the DDE, which is mappable to a stochastic differential equation.  However, we have one more computationally efficient solution method for solving the linear DDE.
 
\subsection{\AD{Diagonalization} by quantum Fourier transform}
\label{ssec: QFT}
In this section we present a quantum algorithm that uses the \AD{QFT} to improve efficiency over \AD{random-walk} methods for solving the DDE. Like the classical FFT approach in \AD{Theorem} \ref{thm: Classical_FFT}, we leverage the QFT to diagonalize a circulant matrix. We show that the key advantage of the QFT is its ability to perform calculations while the quantum state remains in superposition. This enables \AD{diagonalization} in $O(d\log\log n_x)$ \cite{Musk2020AComputations}. Since the QFT efficiently \AD{diagonalize}s $\mathcal{L}$, the only remaining requirement is the implementation of the \AD{nonunitary} operation $\Lambda^k$, where $\Lambda$ is the diagonal matrix of eigenvalues of $\mathcal{L}$ (as in the classical case). In Fig. \ref{fig: quantum_FT} we provide a flowchart illustrating the algorithm's components and their contributions to the overall time complexity.

\begin{figure}
\centering
\includegraphics[]{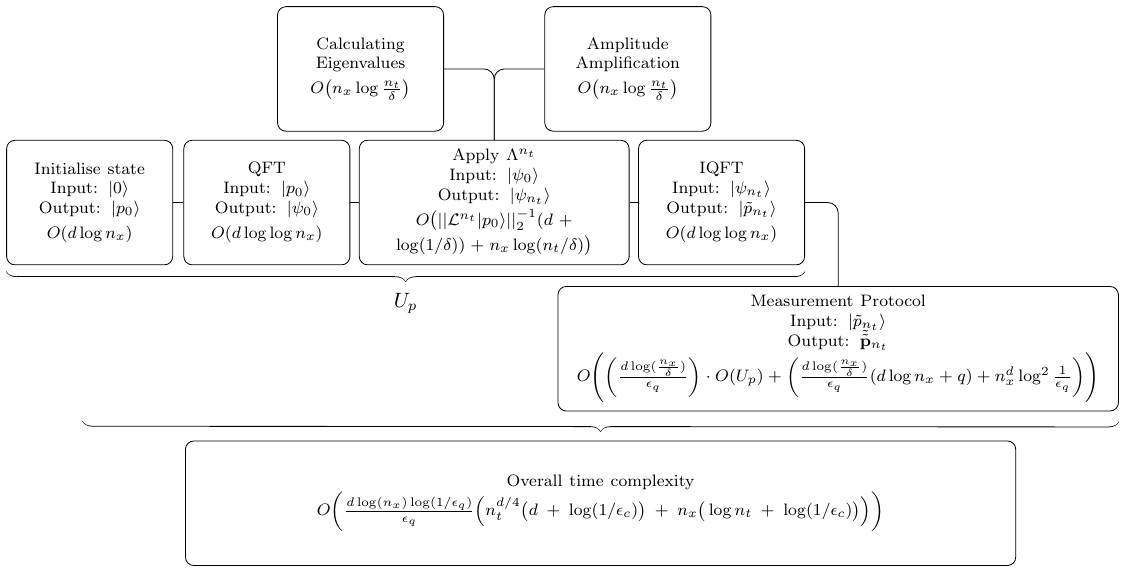}
\caption{A flowchart demonstrating the contributing factors to the overall time complexity of the quantum \AD{diagonalization} solution method for the DDE in \AD{Theorem} \ref{thm: QFT}. The cost of calculating eigenvalues is derived in Lemma \ref{lemma: eigen_L}, the QFT cost is derived in \AD{Ref.} \cite{Musk2020AComputations} and the measurement complexity is derived in \AD{Theorem} \ref{thm: q_meas}.}
    \label{fig: quantum_FT}
\end{figure}

\begin{theorem}[Quantum \AD{diagonalization}] There is a quantum algorithm that estimates the probability distribution $p(\mathbf{x})$ from \AD{Eq.}~\eqref{eqn: DDE} such that  $||\Tilde{\Tilde{p}}(\mathbf{x}, T) - p(\mathbf{x}, T)||_{\infty} \leq \epsilon_c +\epsilon_q $ for all $(\mathbf{x}) \in G$ at a fixed $T=n_t\Delta t$ with $99\% $ success probability in time 
    \begin{equation}
        O\bigg(\frac{d\log(n_x)\log(1/\epsilon_q)}{\epsilon_q} \Big(n_t^{d/4}\big(d + \log(1/\epsilon_c)\big)
        + n_x\big(\log n_t + \log(1/\epsilon_c)\big)\Big) \bigg)
    \end{equation}
    % \begin{equation}
    %      = O\bigg(\frac{d\log\Big(\frac{Td\zeta (aL + D)^2 }{ \epsilon_c L^{d-2}D}\Big)}{\epsilon_q} \bigg(\Big(\frac{Td}{D}\Big)^{d/4}+ \sqrt{\frac{Td\zeta (aL + D)^2}{ \epsilon_c L^{d-2}D}} \bigg) \Big(\log(1/\epsilon_c)+ d\log \Big(\frac{Td\zeta(aL + D)^2}{ \epsilon_c L^{d-2}D}\Big)\Big) \bigg).
    % \end{equation}
    \begin{multline}
         =O\bigg(\frac{d}{\epsilon_q}\log\big(1/\epsilon_q\big)\log\Big(\frac{Td\zeta L^2 (aL + D)^2}{ \epsilon_c D}\Big) 
         \Big(\Big(\frac{T^2d^2\zeta (aL + D)^2}{ \epsilon_c }\Big)^{d/4}\big(d + \log(1/\epsilon_c)\big)\\
         + \sqrt{\frac{Td\zeta L^2(aL + D)^2}{ \epsilon_c D}}\Big(\log\Big(\frac{T^2d^2\zeta (aL + D)^2}{ \epsilon_c }\Big) + \log(1/\epsilon_c)\Big)\Big).
    \end{multline}
  Suppressing the logarithmic terms, this is
    \begin{equation}
        \Tilde{O}\bigg(\frac{d^2n_t^{d/4}}{\epsilon_q} \bigg) = \Tilde{O}\bigg( \frac{d^{(d/2+2)}T^{d/2}\zeta^{d/4} (aL + D)^{d/2}}{ \epsilon_q\epsilon_c^{d/4}}\bigg).
        \label{eqn: qft_comp}
    \end{equation}
    \label{thm: QFT}
\end{theorem}

\begin{proof}
    We begin by preparing the state $|p_0\rangle$, achievable in time $O(d\log n_x)$ (as shown in proof of \AD{Theorem} \ref{thm: QLEM}). Then, we apply an approximate QFT in time $O(d\log n_x \log\log n_x)$ time, yielding the intermediate state $|\psi_0\rangle$ \cite{Musk2020AComputations}.
    Next, as in \AD{Theorem} \ref{thm: Classical_FFT} we apply $\Lambda^{n_t}$ to $|\psi_0\rangle$, where $\Lambda$ is the diagonal matrix corresponding to the eigenvalues of $\mathcal{L}$. The final step to this process is applying the inverse QFT to produce $|\Tilde{p}_{n_t}\rangle$. 
    
  From \AD{Eq.}~\eqref{eqn: eigen_L_full}, the eigenvalues $l_j$ of $\mathcal{L}$ correspond to strings $j = j_1,..., j_d$, where $j_1,...j_d \in {0, ... ,n_x-1}$.
  We can expand~\cite{Linden2022QuantumEquation} 
    \begin{equation}
        |\psi_0\rangle = \sum_{j_1,...,j_d = 0}^{n_x-1} \psi_{j_1,...,j_d}|j_1,...,j_d\rangle.
    \end{equation}
    Then we can apply $\Lambda^{n_t}$ by using an ancilla qubit, performing the map
    \begin{equation}
        |\psi_0\rangle |0\rangle \mapsto \sum_{j_1,...,j_d = 0}^{n_x-1} \psi_{j_1,...,j_d}|j_1,...,j_d\rangle \bigg(l_j^{n_t} |0\rangle + \sqrt{1-l_j^{2n_t}} |0^{\perp}\rangle\bigg)
        \label{eqn: map}
    \end{equation}
    and measuring the ancilla qubit. If, on measurement of the ancilla the outcome is 0, then the state $|\psi_{n_t}\rangle$ is as required.
    Given the classical description of $l_j^{n_t}$ for each $j$, the map in \AD{Eq.} (\ref{eqn: map}) on the ancilla qubit can be achieved up to accuracy $O(\delta)$. This uses $O(d + \log(1/\delta))$ gates and a few additional ancilla qubit which can be reset on completion~\cite{Sanders2018Black-boxArithmetic}.

    The probability that measuring the ancilla qubit provides the required state is $||\mathcal{L}^{n_t}|p_0\rangle||_2^2$.
    We use amplitude amplification, where, $O(||\mathcal{L}^{n_t}|p_0\rangle||_2^{-1})$ repetitions are required to produce the desired state with probability $0.99$. 
    Finally, we apply the inverse QFT to the residual state to produce $|\Tilde{p}_{n_t}\rangle = \mathcal{L}^{n_t}|p_0\rangle / ||\mathcal{L}^{n_t}|p_0\rangle||_2$. 

    To achieve the application of $\Lambda^{n_t}$ we require the eigenvalues of $\mathcal{L}^{n_t}$. We can compute these eigenvalues up to accuracy $\delta$ in time $O(d)$, by using Lemma \ref{lemma: eigen_L}. This is much more efficient than in the classical case in \AD{Theorem}~\ref{thm: Classical_FFT} since all of the eigenvalues can be applied simultaneously in a superposition, with a one off \AD{pre}computation cost of $O(n_x\log(n_t/\delta))$.
    Therefore, we have shown the overall cost of producing the state $|\Tilde{p}_{n_t}\rangle$ is 
    \begin{equation}
        O\Big(d\log n_x + d\log n_x\log\log n_x + ||\mathcal{L}^{n_t}|p_0\rangle||_2^{-1}\big(d + d\log (1/\delta)\big) + n_x\log(n_t/\delta)\Big).
    \end{equation}
    As in the classical case, we take $\delta = \epsilon_c$
    and we also know that $||\mathcal{L}^{n_t}|p_0\rangle||_2^{-1} \leq (4\sqrt{n_t})^{d/2}$ using Lemma \ref{lem: bound L norm}. When implemented this results in
    \begin{equation}
        O\Big(d\log n_x + d\log n_x\log\log n_x + (4\sqrt{n_t})^{d/2}\big(d + \log(1/\epsilon_c)\big)
        + n_x\log n_t + n_x\log(1/\epsilon_c)\Big).
    \end{equation}
    % Focusing on the terms that dominate this expression, we can see that $1/L \propto 1/n_x << n_x$ as well as $\log\log n_x < n_x$ (therefore the QFT step itself is negligible). We also use $\frac{4L\sqrt{n_t}}{n_x} = \sqrt{\frac{Td}{D}}$ as from the proof of \AD{Theorem} \ref{thm: QRW}.
    Upon inspection of Corollary~\ref{cor:dxdt bounds} it is clear that 
        \begin{equation}
            n_t = \frac{n_x^2 TdL^2}{2D}.
        \end{equation}
    Therefore, $dn_t^{d/4}>d\log n_x \log\log n_x$.
    Taking this into account, the overall complexity for the algorithm that produces the final state $|\Tilde{p}_{n_t}\rangle$ is
    \begin{equation}
        O\Big(n_t^{d/4}\big(d + log(1/\epsilon_c)\big)
        + n_x\big(\log n_t + log(1/\epsilon_c)\big)\Big).
    \end{equation}
    
    % Then the overall complexity for the algorithm that produces the final state $|\Tilde{p}_{n_t}\rangle$ is 
    % \begin{equation}
    %     O\Bigg( \bigg(\Big(\frac{Td}{D}\Big)^{d/4}+ n_x \bigg)\big(\log(1/\epsilon_c) + d\log n_x\big)\Bigg).
    % \end{equation}
    As for the previous methods we use \AD{Theorem} \ref{thm: q_meas} to extract the probability distribution $\Tilde{\Tilde{\mathbf{p}}}_{n_t}$ from the state $|\Tilde{p}_{n_t}\rangle$. To achieve this with error probability $\delta$ requires repeating the steps above $O(\log(n_x^d/\delta) \log(1/(\epsilon_q\delta))/\epsilon_q)$.  This uses an additional $O\big(1/\epsilon_q\big(\log (1/(\epsilon_q\delta))(q + \log(1/\epsilon_q))\big)\big)$ gates. Which results in an overall complexity for the final distribution
    \begin{multline}            O\bigg(\log(n_x^d/\delta)\log(1/(\epsilon_q\delta))/\epsilon_q \Big(n_t^{d/4}\big(d + \log(1/\epsilon_c)\big)
        + n_x\big(\log n_t + \log(1/\epsilon_c)\big)\Big)\\
        + \frac{1}{\epsilon_q}\big(\log (1/(\epsilon_q\delta))(q + \log(1/\epsilon_q))\big)\bigg)\AD{.}
    \end{multline}
     In this case $q = O(d\log n_x)$ and $\delta = 0.01$ which gives
    \begin{multline}
        O\bigg(\frac{d\log(n_x)\log(1/\epsilon_q)}{\epsilon_q} \Big(n_t^{d/4}\big(d + \log(1/\epsilon_c)\big)
        + n_x\big(\log n_t + \log(1/\epsilon_c)\big)\Big) + \frac{1}{\epsilon_q}\big(\log (1/(\epsilon_q))(d\log n_x + \log(1/\epsilon_q))\big)\bigg).
    \end{multline}
    % Since $\log(1/\delta) = O(1)$, and 2-qubit gates can be applied in time $O(1)$. 
    Then we can reduce this to 
    \begin{equation}
        O\bigg(\frac{d\log(n_x)\log(1/\epsilon_q)}{\epsilon_q} \Big(n_t^{d/4}\big(d + \log(1/\epsilon_c)\big)
        + n_x\big(\log n_t + \log(1/\epsilon_c)\big)\Big) \bigg)\AD{,}
    \end{equation}
    as the gate complexity of measurement is negligible compared to the larger first term.      
    Finally, we can use Corollary \ref{cor:dxdt bounds}, to substitute in the values for $n_x$ and $n_t$ and write the overall complexity as
    \begin{multline}
         O\bigg(\frac{d}{\epsilon_q}\log\big(1/\epsilon_q\big)\log\Big(\frac{Td\zeta L^2 (aL + D)^2}{ \epsilon_c D}\Big) 
         \Big(\Big(\frac{T^2d^2\zeta (aL + D)^2}{ \epsilon_c }\Big)^{d/4}\big(d + \log(1/\epsilon_c)\big)\\
         + \sqrt{\frac{Td\zeta L^2(aL + D)^2}{ \epsilon_c D}}\Big(\log\Big(\frac{T^2d^2\zeta (aL + D)^2}{ \epsilon_c }\Big) + \log(1/\epsilon_c)\Big)\Big).
    \end{multline}
    Suppressing the logarithmic terms we get
    \begin{equation}
        \Tilde{O}\bigg(\frac{d}{\epsilon_q}\Big(d n_t^{d/4}+n_x\Big) \bigg).
    \end{equation}
    For $d>1$, $d n_t^{d/4} > n_x$. Since this case is most likely for problems where quantum computational advantage is being sought, we choose to drop the final term for clarity. Therefore, our final complexity is
    \begin{equation}
        \Tilde{O}\bigg(\frac{d^2n_t^{d/4}}{\epsilon_q} \bigg) = \Tilde{O}\bigg( \frac{d^{(d/2+2)} T^{d/2}\zeta^{d/4} (aL + D)^{d/2}}{ \epsilon_q\epsilon_c^{d/4}}\bigg).
    \end{equation}
    
    As in the previous cases, we consider the spatial complexity during the measurement step in \AD{Theorem} \ref{thm: q_meas} since $q$ here has the same order of complexity as in the previous method the spatial complexity is also the same 
    \begin{equation}
        O\bigg((1+1/\epsilon_q)\Big(d\log\Big(\frac{Td\zeta L^2(aL+D)^2}{\epsilon_c D}\Big)\Big) + 1/\epsilon_q \log(1/\epsilon_q)\bigg) = \Tilde{O}(d/\epsilon_q).
    \end{equation}
    \end{proof}
In the \AD{diagonalization} by QFT algorithm above, the dominant computational cost arises from two instances of amplitude amplification. The first is used in the measurement process, and the second is involved in constructing the diagonal matrix of eigenvalues. An additional, one-time cost is incurred in \AD{pre}computing the eigenvalues of $\mathcal{L}$, \AD{while} in the classical algorithm described in \AD{Theorem} \ref{thm: Classical_FFT}, the FFT is the primary determinant of computational complexity. We have visually represented the difference in complexity between the quantum and classical approaches in the flowcharts in Fig. \ref{fig: FFT_flow} and Fig. \ref{fig: quantum_FT}.

\AD{Theorem} \ref{thm: QFT} offers the most efficient solution method for \AD{Eq.}~\eqref{eqn: DDE}, this is ensured by bounding $\epsilon_q$. The bound we require ensures that the improvement in efficiency gained in all variables by using \AD{Theorem} \ref{thm: QFT} over \AD{Theorem} \ref{thm: Classical_FFT} outweighs the new dependency introduced by $\epsilon_q$. This must be achieved while maintaining the overall error $\epsilon$. Therefore, quantum computational advantage occurs when 
\begin{equation}
    \frac{1}{\epsilon_q}\leq \Tilde{O}\left(\frac{\zeta^{d/4}L^d(aL+D)^{d/2}}{d \epsilon_c^{d/4}D^{d/2}}\right).
\end{equation}
This method could be extended to all linear PDEs, no matter their order as the matrix will always be circulant and therefore can be \AD{diagonalize}d. This was recently discussed \AD{in Ref.} \cite{Lubasch2025QuantumSpace}. There would need to be further investigation to apply this method to \AD{nonlinear} PDEs with spatially varying parameters. 

\section{Conclusion}
\label{sec: Conclusion}

In this paper we demonstrated a quantum computational advantage for solving the DDE in \AD{Eq.} (\ref{eqn: DDE}).
The magnitude of this advantage depends on the parameter of the problem.
Specifically, for solving the DDE at a fixed time $T$, QFT-based \AD{diagonalization} proved most efficient overall (as discussed Sec. \ref{ssec: QFT}). 
It is an open question beyond the scope of this paper as to whether this could be used for PDEs with spatially varying drift and diffusion terms.

Our algorithm uses \AD{multidimensional} amplitude estimation for probability oracles to extract the probability distribution from the quantum computer, extending previous work that limited these solution methods to expectation value extraction. This introduced a new variable, $\epsilon_q$. 
Quantum advantage is achieved when the complexity is lower whilst maintaining the same total accuracy $\epsilon=\epsilon_c + \epsilon_q$. This requires that the growth in complexity due to $1/\epsilon_q$, is less than the reduction in complexity in the remaining terms of the quantum algorithm when compared to the classical algorithm. Maintaining the same total $\epsilon$ also requires a reduction in $\epsilon_c$ between the classical case and the quantum case to account for the increased $\epsilon_q$ (where $\epsilon_q = 0$ in each of the classical algorithms).

We also extended the quantum linear system, quantum \AD{random-walk}, and QFT methods to solve high-dimensional linear DDEs. 
It revealed that each quantum algorithm outperforms its classical counterpart in terms of time efficiency. 
The quantum \AD{random-walk} method was almost as efficient as quantum diagonalization and is applicable to stochastic PDEs, which offers an opportunity for solving a \AD{nonlinear} DDE. Alternatively, for obtaining solutions at all points in space and time, we proved the quantum linear equations method most efficient. Quantifying the efficiency gains of these extensions remains an open question.

\section*{Acknowledgements}
ED thanks Noah Linden for insights, David Snelling and Sharmila Balamurugan for encouragement and advice and Fujitsu UK Ltd for funding. \\
\\

% \bibliographystyle{apsrev4-2}
% \bibliography{references}

\newpage
\appendix
\section{Ingredients Proofs}
\label{app: strategy}
In this appendix we provide the proofs of \AD{Theorem} \ref{thm: approx linf}.
\setcounter{theorem}{0}
\begin{theorem}[(Restated). Approximation up to $\infty$-norm error] If $\Delta t \leq  \frac{\Delta x^2}{2dD}$, \AD{then}
    \begin{equation}
        ||\Tilde{\mathbf{p}}_{n_t} - \mathbf{p}_{n_t}||_\infty \leq n_t \frac{d \Delta t \zeta}{2L^d}\bigg( d\Delta t (aL+D)^2 + \frac{\Delta x^2}{3} \Big(a L + \frac{D}{2}\Big) \bigg) = T \frac{d \zeta}{2L^d}\bigg( d\Delta t (aL+D)^2 + \frac{\Delta x^2}{3} \Big(a L + \frac{D}{2}\Big) \bigg).
    \end{equation}
\end{theorem}
\begin{proof}
    Using \AD{Eq.} (\ref{eqn: disc OU}) and introducing a new iteration variable, $k \in [0, n_t]$, where $\Tilde{\mathbf{p}}_k$ and $\mathbf{p}_k$ denote the approximate and exact solutions respectively at time $t = k$. Therefore, $\Tilde{\mathbf{p}}_{k+1} = \mathcal{L}\Tilde{\mathbf{p}}_k$ and so $\mathcal{L}$ is stochastic if
    \begin{equation}
        1 - \frac{2dD\Delta t}{\Delta x^2} \geq 0 \mathrm{\quad i.e. \quad} \Delta t \leq \frac{\Delta x^2}{2 dD}\AD{,}
    \end{equation}
    which holds by assumption. Since a stochastic operator must be norm preserving \cite{Voigt1981StochasticEntropy}. Then     
    \begin{multline}
        \bigg|\frac{\Tilde{p}(\mathbf{x}, t + \Delta t) - \tilde{p}(\mathbf{x}, t)}{\Delta t} - \sum_{j = 1}^d \bigg [ \frac{a }{2\Delta x}(\Tilde{p}(...,x_{j[i]} + \Delta x,...,t) - \Tilde{p}(...,x_{j[i]} - \Delta x,...,t)) + \frac{D}{\Delta x^2}(\Tilde{p}(...,x_{j[i]} + \Delta x,...,t)\\
    + \Tilde{p}(...,x_{j[i]} - \Delta x,...,t) - 2\Tilde{p}(\mathbf{x},t))
    \bigg] \bigg|\leq \frac{\Delta t}{2} \max\Bigg|\frac{\partial^2 p}{\partial t^2}\Bigg| + \frac{d a\Delta x^2}{6} \max\Bigg|\frac{\partial^3 p}{\partial x^3}\Bigg| + \frac{ d D \Delta x^2}{12}\max\Bigg|\frac{\partial^4 p}{\partial x^4}\Bigg| 
    \end{multline}
    substituting in the assumptions on the smoothness bounds from Eq.~\eqref{eqn: smoothness_bound} gives
    \begin{equation}
        ||\mathbf{p}_{k+1} - \mathcal{L}\mathbf{p}_k||_{\infty} \leq \Delta t\bigg(\frac{\Delta t}{2} \bigg[ d^2 \zeta(a^2L^2 + 2 a DL + D^2) \bigg] + \frac{d a L \zeta \Delta x^2}{6}  + \frac{d D \zeta \Delta x^2}{12}\bigg)
    \end{equation}
    and simplifying this leads to
    \begin{equation}        
        ||\mathbf{p}_{k+1} - \mathcal{L}\mathbf{p}_k||_{\infty} \leq \frac{d \Delta t \zeta }{2}\bigg(d \Delta t (aL + D)^2 + \frac{\Delta x^2}{3} \Big(a L + \frac{D}{2}\Big)\bigg).
    \end{equation} 
    For clarity this has been simplified by writing $A = (aL + D)^2 $ and $B = a L + \frac{D}{2}$ such that
    \begin{equation}        
        ||\mathbf{p}_{k+1} -\mathcal{L}\mathbf{p}_k||_{\infty} \leq \frac{d \Delta t \zeta}{2}\bigg( d\Delta t A + \frac{\Delta x^2}{3} B \bigg).
    \end{equation} 
    
    Next, we consider the errors introduced by this approach, writing $\Tilde{\mathbf{p}}_k = \mathbf{p}_k + \mathbf{e}_k$ where $\mathbf{e}_k$ is an error vector, then
    \begin{equation}
        \Tilde{\mathbf{p}}_0 = \mathbf{p}_0
    \end{equation}
    and
    \begin{equation}
        \Tilde{\mathbf{p}}_1 =\mathcal{L}\mathbf{p}_0 = \mathbf{p}_1 + \mathbf{e}_1\AD{,}
    \end{equation}
    where $||\mathbf{e}_1||_\infty \leq \frac{d \Delta t \zeta}{2}\bigg( d\Delta t A + \frac{\Delta x^2}{3} B \bigg)$
    \begin{equation}
        \mathbf{\Tilde{p}}_2 =\mathcal{L}\mathbf{\Tilde{p}}_1=\mathcal{L}(\mathbf{p}_1 + \mathbf{e}_1) = \mathbf{p}_2 + \mathbf{e}_2 +\mathcal{L}\mathbf{e}_1\AD{,}
    \end{equation}
    where $||\mathbf{e}_2||_\infty \leq \frac{d \Delta t \zeta}{2}\bigg( d\Delta t A + \frac{\Delta x^2}{3} B \bigg)$ as $\mathcal{L}$ is stochastic, $||\mathcal{L}\mathbf{e}_1||_\infty \leq ||\mathbf{e}_1||_\infty$, so, $||\mathbf{\Tilde{p}}_2 - \mathbf{p}_2||_\infty \leq d \Delta t \zeta\bigg( d\Delta t A + \frac{\Delta x^2}{3} B \bigg)$. Following this argument 
    \begin{equation}
        ||\mathbf{\Tilde{p}}_{n_t} - \mathbf{p}_{n_t}||_\infty \leq n_t \frac{d \Delta t \zeta}{2}\bigg( d\Delta t A + \frac{\Delta x^2}{3} B \bigg) = T \frac{d \zeta}{2}\bigg( d\Delta t A + \frac{\Delta x^2}{3} B \bigg)
    \end{equation}
    as claimed.
\end{proof}
\section{Component eigenvalues}
\label{app: eigenvalues}
This appendix covers the proofs of the Lemmas \ref{lem: eigen_M} and \ref{lem: eigen_AAdag} that are required in the classical technical ingredients section for finding the condition number of $\mathcal{A}$.
% For Theorem 9, with the default counter named "theorem"
\setcounter{theorem}{2}

\begin{lemma}[Restated]
    The eigenvalues of $\mathcal{M}$ are $\{\mu_{j_1}+ ... + \mu_{j_d}: j_1,...,j_d \in \{0,1,...,n_x-1\}\}$, where
    \begin{equation}
        \mu_j = -4\frac{D}{\Delta x^2} \sin^2\left(\frac{\pi j}{n_x}\right) + \frac{a \AD{i}}{\Delta x} \sin\left(\frac{2\pi j}{n_x}\right)
        % \label{eqn: eigen_M}
    \end{equation}
    and $\AD{i}$ denotes the imaginary unit. Moreover, $\mathcal{M}$ is \AD{diagonalize}d by the $d$\AD{th} tensor product of the Fourier transform.
    % \label{lem: eigen_M}
\end{lemma}
\begin{proof}
    Since matrix $\mathcal{M}$ is circulant then it can be \AD{diagonalize}d by the Fourier \AD{t}ransform $\mathcal{F}$.\\
    If $\mathcal{D}_{\mathcal{M}} = \mathrm{diag}\{\mu_0, \mu_1, ..., \mu_{n_x-1}\}$
    is the diagonal matrix that stores the eigenvalues of $\mathcal{M}$, 
    then $\mathcal{D}_{\mathcal{M}} \, \mathcal{F}^{\dag} = \mathcal{F}^{\dag} \, \mathcal{M}$. 
    A circulant matrix has the structure 
    \begin{equation}
        \begin{pmatrix}
            c_0 & c_{n_x - 1} & ... & c_2 & c_1 \\
            c_1 & c_0 & c_{n_x - 1 } & & c_2 \\
            \vdots & c_1 & c_0 & \ddots & \vdots \\
            c_{n_x - 2} & & \ddots & \ddots & c_{n_x - 1} \\
            c_{n_x - 1} & c_{n_x - 2} & ... & c_1 & c_0
        \end{pmatrix}. 
    \end{equation}    
    Therefore, in the case of matrix $\mathcal{M}$, $c_0 = -2\frac{D}{\Delta x^2}$, $c_1 = \frac{D}{\Delta x^2} - \frac{a}{2 \Delta x}$, $c_2 = ... = c_{n-2} = 0$ and $c_{n - 1} = \frac{D}{\Delta x^2} + \frac{a}{2\Delta x}$. Then $\mathcal{D}_{\mathcal{M}} \, \mathcal{F}^{\dag} |0\rangle = \mathcal{F}^{\dag} \, \mathcal{M} |0\rangle$ gives
    \begin{equation}
        \frac{1}{\sqrt{n}}
        \begin{pmatrix}
            \mu_0 \\ \mu_1 \\ \vdots \\ \mu_{n_x - 1}
        \end{pmatrix}
         = \mathcal{F}^{\dag}
         \begin{pmatrix}
            c_0 \\ c_1 \\ \vdots \\ c_{n_x - 1}
        \end{pmatrix}.
    \end{equation}
    If the $j$th eigenvector of the Fourier modes is $\omega_{n_x}^j = e^{\frac{2j\pi \AD{i}}{n_x}}$ then 
    \begin{equation}
        \mu_j = \sum_{k = 0}^{n_x-1} c_k\omega_{n_x}^{-jk} = -2\frac{D}{\Delta x^2} + \bigg(\frac{D}{\Delta x^2} -\frac{a}{2\Delta x}\bigg)\omega_{n_x}^{-j} + \bigg(\frac{D}{\Delta x^2} +\frac{a}{2\Delta x}\bigg)\omega_{n_x}^{-j(n_x - 1)}   
    \end{equation}
    \begin{equation}
         = -2\frac{D}{\Delta x^2} + \frac{2D}{\Delta x^2}\cos\left(\frac{2\pi j}{n_x}\right) + \frac{a \AD{i}}{\Delta x} \sin\left(\frac{2\pi j}{n_x}\right).
    \end{equation}
    \begin{equation}
         = -4\frac{D}{\Delta x^2} \sin^2\left(\frac{\pi j}{n_x}\right) + \frac{2 a \AD{i}}{\Delta x} \sin\left(\frac{\pi j}{n_x}\right)\cos\left(\frac{\pi j}{n_x}\right)
    \end{equation}
    \begin{equation}
         = 2\sin\left(\frac{\pi j}{n_x}\right) \bigg(\frac{a \AD{i}}{\Delta x}\cos\left(\frac{\pi j}{n_x}\right) -\frac{2D}{\Delta x^2} \sin\left(\frac{\pi j}{n_x}\right) \bigg)
    \end{equation}
    as stated.
\end{proof}
\begin{lemma}[Restated.]
    The eigenvalues of $\mathcal{A}_j\,\mathcal{A}_j^{\dag}$ take the form
    \begin{equation}
        \alpha_j = 1 + |l_j|^2 + 2|l_j|\cos \theta = \bigg( \frac{\sin \theta}{\sin (n_t \theta)}\bigg)^2\AD{,}
    \end{equation}
    where $\theta $ is defined by
    \begin{equation}
       |l_j| \sin (n_t \theta) + \sin((n_t+1)\theta) = 0
    \end{equation}
    and $\theta \neq k \pi$ for $k \in \mathbb{N}$.
    % \label{lem: eigen_AAdag}
\end{lemma}
\begin{proof}
The techniques in \AD{Refs.}\cite{Yueh2005EIGENVALUES, Losonczi1992EigenvaluesMatrices} can be used to find a set of eigenvalues $\lambda_k$ for $\mathcal{A}_j\,\mathcal{A}_j^{\dag}$. The starting matrix for this reference is
\begin{equation}
    \mathcal{A}_j\mathcal{A}_j^{\dag} = \begin{pmatrix}
        -\alpha + b & c & 0 & ... & 0 \\
        a & b  & c &  \\
        0 &  a & b  & \ddots & \vdots\\
        \vdots & & \ddots & \ddots & c \\
        0 &  & ... & a  & -\beta + b
    \end{pmatrix}\AD{.}
\end{equation}
In this case $\alpha = |l_j|^2$, $\beta = 0$, $a = -l_j$, $b = 1 + |l_j|^2$ and $c = -l_j^{\dag}$. Then the eigenvalue problem is such that it can be described by the following series of linear equations since $\mathcal{A}_j\,\mathcal{A}_j^{\dag} \mathbf{u} = \lambda \mathbf{u}$
\[u_0 = 0,    
\]
\[
a u_0 + b u_1 + c u_2 = \lambda u_1 + \alpha u_1,
\]
\[
a u_1 + b u_2 + c u_3 = \lambda u_2 + 0,
\]
\[
... = ...,
\]
\[
a u_{n_t-2} + b u_{n_t-1} + c u_{n_t} = \lambda u_{n_t-1} + 0,
\]
\[
a u_{n_t-1} + b u_{n_t} + c u_{n_t+1} = \lambda u_{n_t} + \beta u_{n_t},
\]
\[
u_{n_t+1} = 0,
\]
without repeating the work in \AD{Ref.} \cite{Yueh2005EIGENVALUES} it can be shown that in the case of the matrix $\mathcal{A}_j\,\mathcal{A}_j^{\dag}$ the eigenvalues are $\lambda = 1 + |l_j|^2 + 2|l_j|\cos\theta$ when the following expression is true
\begin{equation}
    -\frac{\sin (n_t + 1) \theta}{\sin (n_t\theta)} = |l_j| = \sqrt{1 - \frac{8 D \Delta t}{\Delta x^2}\sin^2\left(\frac{\pi j}{n_x}\right) + \Delta t^2 \bigg(\frac{16 D^2}{\Delta x^4}\sin^4\left(\frac{\pi j}{n_x}\right) + \frac{a^2}{\Delta x^2}\sin^2\left(\frac{2\pi j}{n_x}\right)\bigg)}.
    \label{eqn: theta_proof}
\end{equation}
and it can be stated that $\theta \neq k \pi$ for $k \in \mathbb{N}$ as the \AD{left-hand side} of \AD{Eq.}~\eqref{eqn: theta_proof} can only be real if $\theta \in \mathbb{R}$. 
\end{proof}

\section{Condition number}
\label{app: cond_num_A}
This appendix provides the proof of \AD{Theorem} \ref{thm: kappa}.
\setcounter{theorem}{4}
\begin{theorem}[Restated.]
    Condition number of $\mathcal{A}$ is
    \[ \kappa = 
        \begin{cases}
            \Theta(n_t) \ if\ \frac{a^2 T}{2n_tD} \leq 1, \\
            \Theta \bigg(\sqrt{n_t^2 + \frac{n_ta^2T}{D}}\bigg) \ if\ \frac{a^2 T}{2n_tD} > 1. 
        \end{cases}
    \]
    Furthermore, 
    \[||\mathcal{A}|| =
    \begin{cases}
        \Theta(1) \ if \ \frac{a^2 \Delta t}{2n_tD} \leq 1, \\
        \Theta\bigg(\sqrt{\frac{a^2 T}{n_tD}}\bigg) \ if\ \frac{a^2 T}{2n_tD} > 1
    \end{cases}
    \]
    and $||\mathcal{A}^{-1}|| = \Theta(n_t)$.
\end{theorem}
\begin{proof}  
The minimum $\sigma_j$ occurs at the minimum $\frac{\sin \theta}{\sin (n_t\theta)}$, $\sin\theta$ is increasing in the interval $\theta \in [0, \pi/2]$ and $|\sin (n_t\theta)|$ is periodic in the interval $\theta \in [0, \pi/n_t]$. 
Furthermore, $\sin (n_t\theta)$ is symmetric around $\theta = \pi/2n_t$ so it is sufficient to consider the interval $\theta \in [0, \pi/2n_t]$. When $\theta$ is small $\sin \theta \geq 2\theta/\pi$ and $\sin (n_t\theta) \leq n_t\theta$. Therefore, 
\begin{equation}
    \sigma_{\AD{\min}} \geq \underset{0< \theta < \Pi}{\bigg|\frac{\sin \theta}{\sin (n_t\theta)}} \bigg| \geq \frac{2}{n_t\pi}
\end{equation}
so $\sigma_{\AD{\min}}  = \Theta (1/n_t)$. 

Next, to estimate $\sigma_{\AD{\max}}$, we consider the maximum $|l_j|$, and $\Delta t \leq \frac{\Delta x^2}{2 dD}$ which when substituted into \AD{Eq.} (\ref{eqn: theta_proof}) means $|l_j|$ can be written as
\[
    |l_j| \leq \sqrt{1 - 4\sin^2\bigg(\frac{\pi j}{n_x}\bigg) + 4\sin^4\bigg(\frac{\pi j}{n_x}\bigg) + \frac{a^2 \Delta t}{2D}\sin^2\bigg(\frac{2\pi j}{n_x}\bigg)}
\]
and then simplified to
\begin{equation}
    |l_j| \leq \sqrt{1 + \bigg(\frac{a^2 T}{2n_tD} -1\bigg) \sin^2\bigg(\frac{2\pi j}{n_x}\bigg)}.       
\end{equation}
Therefore, the maximum of $|l_j|$ is the maximum of whichever of these parts is larger as follows 
\[|l_j|_{\AD{\max}} = 
\begin{cases} 
    1\ \mathrm{if}\ \frac{a^2 T}{2n_tD} \leq 1\ \mathrm{at}\ j = n_x/2, \\
    \sqrt{\frac{a^2 T}{2n_tD}}\ \mathrm{if}\ \frac{a^2 T}{2n_tD} > 1\ \mathrm{at}\ j = n_x/16
\end{cases}
\]
resulting in
\[\sigma_{\AD{\max}} \leq 
\begin{cases} 
    2\ \AD{\mathrm{if}}\ \frac{a^2 \Delta t}{2D} \leq 1, \\
    \sqrt{1+ \frac{a^2 T}{2n_tD} + \sqrt{\frac{2a^2 T}{n_tD}}}\ \AD{\mathrm{if}}\ \frac{a^2 T}{2n_tD} > 1.
\end{cases}
\]
Given the condition number is calculated as $\kappa_A = \AD{\frac{\sigma_{\max}}{\sigma_{\min}}}$, it can be written as
\begin{equation}
    \kappa_A = 
    \begin{cases}
        2 n_t \ \AD{\mathrm{if}}\ \frac{a^2 T}{2n_tD} \leq 1, \\
        n_t\sqrt{1+ \frac{a^2 T}{2n_tD} + \sqrt{\frac{2a^2 T}{n_tD}}} = n_t \sqrt{1 + \frac{a^2T}{2n_tD} + a\sqrt{\frac{2 T}{n_tD}}} \ \AD{\mathrm{if}}\ \frac{a^2 T}{2n_tD} > 1. 
    \end{cases}
\end{equation}
and so it can be stated
\begin{equation}
    \kappa_A = 
    \begin{cases}
        \Theta(n_t) \ \AD{\mathrm{if}}\ \frac{a^2 T}{2n_tD} \leq 1, \\
        \Theta\Big(\sqrt{n_t^2 + \frac{n_ta^2T}{D}}\Big) \ \AD{\mathrm{if}}\ \frac{a^2 T}{2n_tD} > 1. 
    \end{cases}
\end{equation}

\end{proof}
Both results have been included here for completeness but in most cases $\kappa_A = \Tilde{O}(n_t)$ can be used since $T = n_t\Delta t$. So, consider $a^2\Delta t/2D >1$ which is unlikely to be very large as $\Delta t$ is inherently small, which means that $\kappa_A = \sqrt{n_t^2 + n_ta^2T/D} = \sqrt{n_t^2 + n_t^2a^2\Delta t/D} \approx n_t$.

\section{Bounding the differential operator}
\label{app: bound_l_lemma}
This appendix provides the proofs to Lemma~\ref{lem: kappa_l} and Lemma \ref{lem: bound L norm}. 
\setcounter{theorem}{10}
\begin{lemma}[Restated: Condition number of $\mathcal{L}$] 
    The condition number of $\mathcal{L}$ from \AD{Eq.}~\eqref{eqn: M to L} is $\kappa_{\mathcal{L}} = 5$ when $\frac{D\Delta t}{\Delta x^2} \leq 1/5$ and $a/D<2\sqrt{10}$. 
\end{lemma}
\begin{proof}      
    Since $\mathcal{L}\,\mathcal{L}^{\dagger}$ is circulant, repeat the steps taken in the proof of Lemma \ref{lem: eigen_AAdag} to find the eigenvalues.

    $\mathcal{L}$ is circulant with structure $c_0 = 1 - \frac{2D\Delta t}{\Delta x^2}$, $c_1 = \Delta t \big(\frac{D}{\Delta x^2}- \frac{a}{2\Delta x}\big)$, $c_{n_x-1} = \Delta t \big(\frac{D}{\Delta x^2} + \frac{a}{2\Delta x}\big)$ and $c_2 = ... = c_{n_x-2} = 0$. Therefore, $\mathcal{L}\,\mathcal{L}^{\dagger}$ is circulant with structure $d_0 = c_0^2 + c_1^2 + c_2^2 = (1 - \frac{2D\Delta t}{\Delta x^2})^2 + 2\Delta t^2(\frac{a^2}{4\Delta x^2} + \frac{D^2}{\Delta x^4})$, $d_1 = d_{n_x-1} = c_0(c_1+c_2) = \frac{2D\Delta t}{\Delta x^2}(1 - \frac{2D\Delta t}{\Delta x^2})$, $d_2 = d_{n_x-2} = c_1 c_{n_x-1} = \frac{D^2\Delta t^2}{\Delta x^4} + \frac{a^2\Delta t^2}{4\Delta x^2}$ and $d_3 =... = d_{n_x-3} = 0$.
    This results in 
    \begin{multline}                        \lambda_j(\mathcal{L}\,\mathcal{L}^{\dagger}) = \bigg(1 - \frac{2D\Delta t}{\Delta x^2}\bigg)\bigg(1 + \frac{2D\Delta t}{\Delta x^2}\bigg(2\cos\bigg(\frac{2\pi j}{n}\bigg)-1\bigg)\bigg) + \frac{4D^2\Delta t^2}{\Delta x^4}\cos^2\bigg(\frac{2\pi j}{n_x}\bigg) +\frac{a^2\Delta t^2}{\Delta x^2}\sin^2\bigg(\frac{2\pi j}{n_x}\bigg).
        \label{eqn: lambda_L}
    \end{multline}
There are a few different regimes to consider when looking for the minimum and maximum eigenvalues. This may be easier to consider by rewriting \AD{Eq.} \ref{eqn: lambda_L} with new variables
\begin{equation}
    \lambda_j(\mathcal{L}\,\mathcal{L}^{\dagger}) = (1 - 2Y)((1 -2Y) + 4Y\cos(Z))) + 4Y^2\cos^2(Z) + 4X^2 \sin^2(Z)\AD{,}
    \label{eqn: lambda_L_YZ}
\end{equation}
where $Y = \frac{D\Delta t}{\Delta x^2}$, $X = \frac{a \Delta t}{2 \Delta x^2} = \frac{a}{2D}Y$ and $Z = \frac{2\pi j}{n_x}$. 
We can bound these variables as follows; $0 < Y \leq 1/2d \leq 1/2$  based on Theorem \ref{thm: approx linf}, $0<X\leq a/4D$ and $0 \leq Z \leq 2\pi$. However, on inspection at $Y = 1/4$ the eigenvalue finds a minimum at 
$\lambda(\mathcal{L}\mathcal{L}^{\dag})_{\AD{\min}} = 0$. This would cause the condition number $\kappa \rightarrow 0$. To force a \AD{noninfinite} condition number we bound the variables $X$ and $Y$ instead to, $0 < Y \leq 1/5$ and $0 < X \leq a/10D$. This is true for all $d \geq 3$, and is only a single order of magnitude reduction. 

To begin with we
consider when $Z = \pi$, to find the minimum. Substituting this into \AD{Eq.} \ref{eqn: lambda_L_YZ} gives
\begin{equation} 
    \lambda_j(\mathcal{L}\,\mathcal{L}^{\dagger}) = (1-2Y)(1-6Y)+4Y^2.
\end{equation}
Then we can find the minimal result at the maximum $Y = 1/5$, therefore, $\lambda_{\AD{\min}}(\mathcal{L}\,\mathcal{L}^{\dagger}) = 1/25$. To find the maximum there are 2 scenarios. When either $Y=0$ or $Z=0$ it is easy to see that $\lambda_{\AD{\max}}(\mathcal{L}\,\mathcal{L}^{\dagger}) = 1$. 

Using these results we can find the condition number $\kappa_{\mathcal{L}}$ as
\begin{equation}
    \kappa_{\mathcal{L}} = \sqrt{\frac{\lambda_{max}}{\lambda_{min}}} = 5\AD{.}
\end{equation}

There is also a case at very small $Z$ e.g. $Z = \pi/64$, (where $n_x \geq 128$ where a large $X^2$ has a dominant factor. In this case the increase in the term $X^2\sin^2Z$ is greater than the reduction in the $\cos Z$ terms. This only occurs at $a/D > 2\sqrt{10}$, so we only consider the region for which the maximal eigenvalue is well defined as stated. 
\end{proof}

\begin{lemma}[Restated: Norm. of $\mathcal{L}$] 
    Let $\mathcal{L}$ be defined by \AD{Eq.}~\eqref{eqn: diff operator}, taking $\Delta t = \Delta x^2 /(2dD)$ as in Corollary \ref{cor:dxdt bounds}. Then for any integer $\tau \geq 1$,
    \begin{equation}
       \frac{1}{(4\sqrt{\tau})^d} \leq \langle 0|\mathcal{L}^{2\tau}|0\rangle = ||\mathcal{L}^{\tau}||_2^2.
    \end{equation}
\end{lemma}
\begin{proof}
    Since $\mathcal{L}$ describes a simple \AD{random-walk} on a periodic $d$-dimensional square lattice and 
    \begin{equation}
        ||\mathcal{L}^{\tau}|0\rangle||_2^2 = \langle 0|\mathcal{L}^{2\tau}|0\rangle,
    \end{equation}
    where $|0\rangle$ denotes the origin, then  $||\mathcal{L}^{\tau}|0\rangle||_2^2$ can be interpreted as the probability of returning to the origin after $2\tau$ steps of the \AD{random-walk}. 
    
    The next step is to bound this quantity, first consider the $d = 1$ case.
    A simple lower bound follows by observing that the probability of returning to $0$ after $2\tau$ steps is lower bounded by the probability of a \AD{random-walk} on the integers returning to $0$ after $2\tau$ steps which is exactly $(_{\tau}^{2\tau})/2^{2\tau}$, such that
    \begin{equation}
        \langle 0|\mathcal{L}^{2\tau}|0\rangle \geq \frac{(_{\tau}^{2\tau})}{2^{2\tau}} \geq \frac{1}{2\sqrt{\tau}}.
    \end{equation}
    
    We can extend this bound to the $d$- dimensional case. Based on the interpretation of $\langle 0|\mathcal{L}^{2\tau}|0\rangle$ as the probability of returning to the origin after $2\tau$ steps of a \AD{random-walk}. Each step corresponds to choosing one of $d$ dimensions uniformly at random. Then moving in one of 2 possible directions within that dimension. Since the lower bound for this in one dimension is $\frac{1}{2\sqrt{\tau}}$, then if each of the $d$ independent \AD{random-walk}s make an even number of steps, the probability that they all simultaneously return to the origin is at least $\frac{1}{(2\sqrt{\tau})^d}$. For this to be true, we must lower bound the probability that all of the walks take an even number of steps. This is demonstrated by \AD{Ref.} \cite{Linden2022QuantumEquation} and recreated here. 

    Let $N_e(d, 2\tau)$ denote the number of sequences of $2\tau$ integers between 1 and $d$ such that the number of times each integer appears in the sequence is even. The probability that all the walks take an even number of steps is $N_e(d, 2\tau)/d^{2\tau}$. Then we can show by induction on $d$ that $N_e(d, 2\tau)\geq d^{2\tau}/2^d$. For the best case, $N_e(1, 2\tau) = 1 \geq 1/2$ as required. Then for $d\geq2$,
    \begin{equation}
        N_e(d, 2\tau) = \sum_{j=0}^{\tau}\binom{2\tau}{2j} N_e(d-1, 2\tau - 2j)
        \geq \sum_{j=0}^{\tau}\binom{2\tau}{2j}\frac{1}{2^d-1}(d-1)^{2\tau-2j}        = (d-1)^{2\tau}\frac{1}{2^d-1} \sum_{j=0}^{\tau}\binom{2\tau}{2j}(d-1)^{-2j}\AD{,}
    \end{equation}
    \begin{equation}
        (d-1)^{2\tau}\frac{1}{2^d-1}\frac{1}{2}\bigg(\bigg(1+\frac{1}{d-1}\bigg)^{2\tau}+\bigg(1-\frac{1}{d-1}\bigg)^{2\tau}\bigg)\AD{,}
    \end{equation}
    \begin{equation}
        \frac{1}{2^d}(d^{2\tau}+(d-2)^{2\tau}) \geq \frac{1}{2^d}d^{2\tau}.
    \end{equation}
    Therefore, with probability at least $1/2^d$, all of the walks make an even number of steps, and the probability that they all return to the origin after $2\tau$ steps in total is at least $1/(4\tau)^d$. 
\end{proof}

\end{document}